\documentclass[a4paper]{llncs}
\usepackage{color}
\usepackage{xcolor}
\usepackage{amsmath}
\usepackage{amssymb} 
\usepackage{graphicx}
\usepackage{xspace}
\usepackage{subfig}

\newcommand{\RED}[1]{\textcolor{red}{#1}}
\newcommand{\BROWN}[1]{\textcolor{brown}{#1}}

\newcommand{\MAGENTA}[1]{\textcolor{magenta}{#1}}

\newcounter{dgnot} 
\newenvironment{dgnot}[1][]{\refstepcounter{dgnot}\par\medskip
   \noindent \textbf{\RED{NfD~\thedgnot.}  #1} \rmfamily}{\medskip}

\newcounter{mknot} 
\newenvironment{mknot}[1][]{\refstepcounter{mknot}\par\medskip
   \noindent \textbf{\MAGENTA{NfM~\themknot.}  #1} \rmfamily}{\medskip}

\newcounter{vcnot} 
\newenvironment{vcnot}[1][]{\refstepcounter{vcnot}\par\medskip
   \noindent \textbf{\BROWN{NfV~\thevcnot.}  #1} \rmfamily}{\medskip}

\newcounter{eknot} 
\newenvironment{eknot}[1][]{\refstepcounter{vcnot}\par\medskip
   \noindent \textbf{\BROWN{NfE~\thevcnot.}  #1} \rmfamily}{\medskip}

\newcommand{\calC}{\mathcal{C}}

\newcommand{\calF}{\mathcal{F}}

\newcommand{\calI}{\mathcal{I}}
\newcommand{\calJ}{\mathcal{J}}

\newcommand{\calL}{\mathcal{L}}
\newcommand{\calM}{\mathcal{M}}
\newcommand{\calN}{\mathcal{N}}

\newcommand{\calP}{\mathcal{P}}

\newcommand{\calS}{\mathcal{S}}

\newcommand{\calV}{\mathcal{V}}

\newcommand{\theo}{Theorem}
\newcommand{\propo}{Proposition}
\newcommand{\lm}{Lemma}

\newcommand{\deriv}{\noindent\hspace*{.20in}\vspace{0.1in}}

\newcommand{\contnl}{\\$ $\mbox{\hspace{0.2in}}}

\newcommand{\hint}[2]{\newline#1\hspace*{.25in} [\,$#2$\,]\vspace{0.1in}\newline\hspace*{.20in}
\vspace{0.1in}}

\newcommand{\SET}[1]{\{#1\}}
\newcommand{\ZET}[2]{\SET{#1 \,|\, #2}}
\newcommand{\pws}{\mathbf{\calP}}

\newcommand{\nats}{\mathbb{N}}
\renewcommand{\succ}{\mathtt{succ}}
\newcommand{\reals}{\mathbb{R}}

\newcommand{\cs}{CS}
\newcommand{\qdcs}{QdCS}
\newcommand{\closure}{\calC}
\newcommand{\closureF}{\stackrel{\rightarrow}{\calC}}
\newcommand{\closureT}{\stackrel{\leftarrow}{\calC}}
\newcommand{\interior}{\calI}
\newcommand{\interiorF}{\stackrel{\rightarrow}{\calI}}
\newcommand{\interiorT}{\stackrel{\leftarrow}{\calI}}
\newcommand{\pthlen}{\mathtt{len}}
\newcommand{\dom}{\mathtt{dom}}
\newcommand{\rng}{\mathtt{range}}
\newcommand{\pths}{\mathtt{Paths}}
\newcommand{\bpths}{\mathtt{BPaths}}
\newcommand{\bpthsF}{\mathtt{BPaths^F}}
\newcommand{\bpthsT}{\mathtt{BPaths^T}}
\newcommand{\leadspth}[1]{\stackrel{#1}{\Longrightarrow}}
\newcommand{\T}{\mathtt{Tr}}

\newcommand{\dbs}{dbs}
\newcommand{\dbsc}{dbsc}
\newcommand{\cl}{$\calC$}
\newcommand{\cm}{CM}
\newcommand{\cmc}{CMC}
\newcommand{\pth}{Path}
\newcommand{\cop}{CoPa}
\newcommand{\qdcm}{QdCM}
\newcommand{\ap}{{\tt AP}}
\newcommand{\model}{\calM}
\newcommand{\peval}{\calV}
\newcommand{\invpeval}{\peval^{-1}}

\newcommand{\eqsign}{\rightleftharpoons}
\newcommand{\cmbis}{\eqsign_{\mathtt{\cm}}}
\newcommand{\inlbis}{\eqsign_{\mathtt{\inl}}}
\newcommand{\homeo}{\eqsign_{\mathtt{HO}}}
\newcommand{\cmcbis}{\eqsign_{\mathtt{\cmc}}}
\newcommand{\clbis}{\eqsign_{\calC}}
\newcommand{\pthbis}{\eqsign_{\mathtt{Pth}}}
\newcommand{\copbis}{\eqsign_{\mathtt{\cop}}}
\newcommand{\treq}{\eqsign_{\mathtt{Tr}}}
\newcommand{\apeq}{\eqsign_{\ap}}
\newcommand{\dbseq}{\eqsign_{\mathtt{\dbs}}}
\newcommand{\dbsceq}{\eqsign_{\mathtt{\dbsc}}}

\newcommand{\iml}{{\tt IML}}
\newcommand{\inl}{{\tt INL}}
\newcommand{\imlc}{{\tt IMLC}}
\newcommand{\irl}{{\tt IRL}}
\newcommand{\icrl}{{\tt ICRL}}
\newcommand{\slcs}{{\tt SLCS}}
\newcommand{\islcs}{{\tt ISLCS}}
\newcommand{\minilogica}{{\tt MiniLogicA}}
\newcommand{\topochecker}{{\tt topochecker}}
\newcommand{\voxlogica}{{\tt VoxLogicA}}
\newcommand{\form}{\Phi}
\newcommand{\ltrue}{{\tt true}}
\newcommand{\lfalse}{{\tt false}}
\newcommand{\lneg}{\neg}
\newcommand{\lior}{\bigvee}
\newcommand{\liand}{\bigwedge}
\newcommand{\lnear}{\calN}
\newcommand{\lsurr}{\calS}
\newcommand{\lprop}{\calP}
\newcommand{\lnearF}{\stackrel{\rightarrow}{\calN}}
\newcommand{\lnearT}{\stackrel{\leftarrow}{\calN}}
\newcommand{\ltothru}{\stackrel{\rightarrow}{\rho}}
\newcommand{\lfromthru}{\stackrel{\leftarrow}{\rho}}
\newcommand{\lto}{\stackrel{\rightarrow}{\sigma}}
\newcommand{\lfrom}{\stackrel{\leftarrow}{\sigma}}
\newcommand{\lstothru}{\stackrel{\rightarrow}{\zeta}}
\newcommand{\lsfromthru}{\stackrel{\leftarrow}{\zeta}}

\newcommand{\imleq}{\simeq_{\iml}}
\newcommand{\irleq}{\simeq_{\irl}}
\newcommand{\icrleq}{\simeq_{\icrl}}
\newcommand{\imlceq}{\simeq_{\imlc}}
\newcommand{\islcseq}{\simeq_{\islcs}}

\newcommand{\closedefi}{\hfill$\bullet$}
\newcommand{\closethm}{}

\newcommand{\closerem}{}

\newcommand{\sep}{\vert}
\newcommand{\sem}[1]{[\![#1]\!]}

\newcommand{\wlg}{without loss of generality}

\begin{document}

\mainmatter



\title{On Bisimilarities for Closure Spaces\\Preliminary Version
\thanks{Research partially supported by the MIUR Project PRIN 2017FTXR7S IT-MaTTerS".
The authors are listed in alphabetical order, as they equally contributed to
this work.}}

\author{Vincenzo~Ciancia\inst{1} \and Diego~Latella\inst{1}  \and Mieke~Massink\inst{1}, Erik de Vink\inst{2}}

\institute{
CNR-ISTI, Pisa, Italy,\\
\email{$\{$V.Ciancia, D.Latella, M.Massink$\}$@cnr.it}
\and
Univ. of Eindhoven\\
\email{evink@win.tue.nl}}

\authorrunning{Ciancia et al.}

\maketitle

\begin{abstract}
Closure spaces are a generalisation of topological spaces obtained by removing the  idempotence requirement on the closure operator. We adapt the standard notion of bisimilarity for topological models, namely Topo-bisimilarity,  to closure models---we call the resulting equivalence {\em \cm-bisimilarity}---and refine it for quasi-discrete closure models. We also define two additional notions of bisimilarity that are based on paths on space, namely {\em \pth-bisimilarity} and {\em Compatible \pth-bisimilarity}, {\em \cop-bisimilarity} for short. The former expresses (unconditional) reachability, the latter refines it in a way that is reminishent of {\em Stuttering Equivalence} on transition systems. For each bisimilarity we provide a logical characterisation, using variants of \slcs{.} We also address the issue of (space) minimisation via the three equivalences.
\end{abstract}


\section{Introduction}\label{sec:Introduction}

In the well known topological interpretation of model logic a point in space  satisfies formula $\diamond \, \form$ whenever it belongs to the {\em topological closure} of the set $\sem{\form}$ of all the points satisfying formula $\form$ (see e.g.~\cite{vBB07}). Topological spaces form the fundamental basis for reasoning about space, but the idempotence property of topological closure turns out to be too restrictive. For instance, discrete structures useful for certain representations of space, like general graphs, cannot be captured. To that purpose, a more liberal notion of space, namely that of {\em closure spaces}, has been proposed in the literature that does not require  idempotence of the closure operator (see~\cite{Gal03} for an in-depth treatment of the subject). 

In~\cite{Ci+14,Ci+16} the \emph{Spatial Logic for Closure Spaces} ($\slcs$) has been proposed that enriches modal logic with a {\em surrounded} operator $\lsurr$ such that a point $x$ satisfies
$\form_1 \, \lsurr \, \form_2$ if it lays in a set $A \subseteq \sem{\form_1}$ and the external border of which is composed by points in $\sem{\form_2}$, i.e. $x$ satisfies $\form_1$ and is surrounded by points satisfying $\form_2$. A model checking algorithm has been proposed in 
~\cite{Ci+14,Ci+16}  that has been implemented in  the tool \topochecker~\cite{Ci+18,Ci+15} and, more recently, in \voxlogica~\cite{Be+19}, a tool specialised for spatial model-checking digital images, that can be modelled as {\em adjacency spaces}, a special case of closure spaces.

The logic and its model checkers have been applied to several case
studies~\cite{Ci+16,Ci+15,Ci+18} including a declarative approach to
medical image analysis~\cite{Be+19,Be+19a,Ba+20,Be+21}. An encoding of the
discrete Region Connection Calculus RCC8D of~\cite{Ra+13} into the collective
variant of \slcs{} has been proposed in~\cite{Ci+19b}. The logic has also
inspired other approaches to spatial reasoning in the context of signal temporal
logic and system monitoring~\cite{BBLN17,NBCLM18} and in the verification of cyber-physical systems \cite{TKG17}.
In~\cite{Li+20} it has been shown that \slcs{} cannot express topological separation and connectedness; the authors propose a notion of {\em path preserving bisimulation}.

A key question, when reasoning about modal logics and their models, is the relationship between  logical equivalences and notions of bisimilarity defined on their underlying models. This is also important because the existence of such bisimilarities, and their logical characterisation, makes it possible to exploit minimisation procedures for  bisimilarity for the purpose of efficient model-checking. 

In this paper we study three different notions of bisimilarity for closure models, i.e.
models based on closure spaces. 
The first one is {\em \cm-bisimilarity}, that is an adaptation for closure models of classical Topo-bisimilarity for topological models~\cite{vBB07}. Actually,  \cm-bisimilarity is an instantiation to closure models of Monotonic bisimulation on neighbourhood models~\cite{vB+17,Han03}. In fact,  it is defined
using the interior operator of closure models, that is monotonic, thus making closure models an intantiation of monotonic neighbourhood models.
We show that \cm-bisimilarity is weaker than homeomorphism and provide a logical 
characterisation of the former, namely the Infinitary Modal Logic.

We then present a refinement of
 \cm-bisimilarity, specialised  for {\em quasi-discrete} closure
models, i.e. closure models where every point has a minimal neighbourhood. In this case, the closure of a set of points---and so also its interior---can be expressed using an underlying binary relation; this gives raise to both a {\em direct} closure and interior of a set, and a {\em converse} closure and interior, the latter being obtained using the inverse of the  binary relation. This, in turn, induces a refined notion of bisimilarity, {\em \cm-bisimilarity with converse},
which, on quasi-discrete closure models, is shown to be strictly stronger than \cm-bisimilarity. We also introduce a notion of {\em Trace Equivalence} for closure models and show that \cm-bisimilarity with converse implies Trace Equivalence, but not the other way around. 

We extend the Infinitary Modal Logic with the converse of its unary modal operator and show that the resulting logic characterises \cm-bisimilarity with converse.
\cm-bisimulation with converse, as \cm-bisimulation, is defined using the {\em interior} operator, $\interior$.
We show that \cl-bisimulation, proposed in~\cite{Ci+20}, and resembling Strong Back-and-Forth bisimilarity for processes proposed in~\cite{De+90}, coincides with \cm-bisimulation with converse. The definition of  \cl-bisimulation uses the  {\em closure} operator $\closure$, i.e. the dual of $\interior$. The advantage of using directly the closure operator, which is
the foundational operator of closure spaces, is given by its intuitive interpretation
in quasi-discrete closure model{s} that makes several proofs simpler. We recall here that in~\cite{Ci+20} a minimisation algorithm for \cl-bisimulation, and related tool, \minilogica, have been proposed as well.
We show that the infinitary extension \islcs{} of (a variant of) \slcs,  fully characterises \cl-bisimulation. The variant of \slcs{} of interest here is the one with two modal operators expressing {\em (conditional) reachability}. More specifically, one operator expresses the possibility that a point in space {\em may reach} an area satisfying a given formula\footnote{By ``area satisfying'' here we mean ``all the points of which satisfy''.} $\form_1$ via a path the points of which satisfy a formula $\form_2$; the other expresses the possibility that a point in space {\em may be reached} from an area satisfying a given formula $\form_1$ via a path the points of which satisfy a formula $\form_2$. The classical Infinitary Modal Logic modal operator, and its converse, can be derived from the reachability operators of \slcs{,} when the underlying model is quasi-discrete\footnote{We also show that, for {\em general} \cm{,} the {\em surrounded} operator of \slcs{} can be derived from the reachability ones.}. This last result brings to the
coincidence of \cm-bisimilarity with converse and \cl-bisimilarity for quasi-discrete closure models.

\cm-bisimilarity, and \cm-bisimilarity with converse, play an important role  as they are the  counterpart of classical Topo-bisimilarity. On the other hand,
they turn out to be rather too strong when one has in mind intuitive relations on space like, e.g. scaling, that may be useful when dealing with models representing images 
(see~\cite{Ba+20,Be+21,Be+19,Be+19a} for details).
For this purpose, we introduce our second, weaker notion of bisimilarity,
namely  \pth-bisimulation that is essentially based on {\em reachability} of bisimilar points by means of paths over the underlying space. We show that, for quasi-discrete closure models,
 \pth-bisimilarity is strictly weaker than \cm-bisimilarity with converse; we also show 
 that a similar result does  {\em not} hold for general \cm-bisimilarity and general closure models.  We provide a remedy to such problem, for the case in which the space is path-connected, using an adaptation for \cm{s} of \inl-bisimilarity~\cite{vB+17}. 
We furthermore show that \pth-bisimilarity and Trace Equivalence for general \cm{s} are uncomparable. We finally consider the Infinitary Modal Logic where the modal operator is replaced by two unary modalities---one for (unconditional) reachability {\em of} an area satisfying a given formula, and the other for (unconditional) reachability {\em from} an area satisfying a given formula---and prove that such a logic characterises \pth-bisimilarity.

\pth-bisimilarity is in some sense too weak, abstracting too much; nothing whatsoever is required of the relevant paths, except their starting points being fixed and related by the bisimulation, and  their end-points be in the bisimulation as well. A little bit deeper insight into the structure of such paths would be desirable as well as some, relatively high level, requirements on them. To that purpose we resort to a notion of ``compatibility'' between relevant paths that essentially requires each of them to be composed by a sequence of non-empty ``zones'', with the total number of zones in each of the two paths being the same, while the length of each zone being arbitrary (but at least 1); 
each element of one path in a  given zone is required to be related by the bisimulation to all the elements in the corresponding zone in the other path.
This idea of compatibility gives rise to the third notion of bisimulation, namely {\em Compatible Path bisimulation}, \cop-bisimulation, which is strictly stronger than \pth-bisimilarity and, for quasi-discrete closure model{s}, strictly weaker than \cm-bisimilarity with converse. We also show that Compatible Path bisimulation and
Trace Equivalence are uncomparable and we provide a logical characterisation of Compatible Path bisimulation using a restricted version of \islcs{.}
The notion of \cop-bisimulation is reminiscent of that of the {\em Equivalence with respect to Stuttering} for transition systems proposed in~\cite{BCG88}, although in a  different context and with different definitions as well as underlying notions. 


The  paper is organised as follows:
after having settled the context and offered some preliminary notions and definitions in Section~\ref{sec:Preliminaries}, in Section~\ref{sec:CMbisimilarity} we present \cm-bisimilarity. Section~\ref{sec:CMCbisimilarity} deals with \cm-bisimulation with converse.
Section~\ref{sec:Pathbisimilarity} addresses \pth-bisimilarity, while in 
Section~\ref{sec:CoPabisimilarity} \cop-bisimilarity is dealt with.
We conclude the paper with Section~\ref{sec:Conclusions}.
All detailed proofs are provided in the Appendix.
\section{Preliminaries}\label{sec:Preliminaries}

In this paper, given set $X$, $\pws(X)$ denotes the powerset of $X$;
for $Y \subseteq X$ we let $\overline{Y}$ denote $X\setminus Y$, i.e. the complement of $Y$. For function $f:X \to Y$ and
 $A \subseteq X$, we let $f(A)$ be defined as $\ZET{f(a)}{a \in A}$.
For binary relation $R \subseteq X \times X$ we let
$R^{-1}$ denote the relation  $\ZET{(x_1,x_2)}{(x_2,x_1)\in R}$, whereas 
$R^{r}$ ($R^{s}$, respectively)  will denote the {\em transitive  closure} ({\em symmetric closure}, respectively) of $R$, and
$R^{rst}$ will denote the {\em reflexive}, {\em symmetric} and {\em transitive} closure 
of $R$. In this section, we recall several definitions and results on closure spaces, most of which are taken from~\cite{Gal03}.

\begin{definition}[Closure Space - \cs]\label{def:ClosureSpace}
A {\em closure space}, \cs{} for short, is a pair $(X,\closure)$ where $X$ is a non-empty set (of {\em points})
and {\em $\closure: \pws(X) \to \pws(X)$} is a function satisfying the following axioms:
\begin{enumerate}
\item $\closure(\emptyset)=\emptyset$;
\item $A \subseteq \closure(A)$ for all  $A \subseteq X$;
\item $\closure(A_1 \cup A_2) = \closure(A_1) \cup \closure(A_2)$ for all $A_1,A_2\subseteq X$.
\closedefi
\end{enumerate}
\end{definition}

It is worth pointing out that topological spaces coincide with the sub-class of \cs{s} for which also the {\em idempotence} axiom $\closure(\closure(A) = \closure(A)$ holds. 
The {\em interior} operator is the dual of closure: $\interior(A)=\overline{\closure(\overline{A})}$. A {\em neighbourhood} of a point $x \in X$  is any set $A \subseteq X$ such that $x \in \interior(A)$. 
A minimal neighbourhood of a point $x$ is a neighbourhood $A$ of $x$ such that
$A \subseteq A'$ for any other neighbourhood $A'$ of $x$. 

We recall here the fact that the {\em closure} and, consequently,
the {\em interior} operators are monotonic: if $A_1 \subseteq A_2$ then
$\closure(A_1) \subseteq \closure(A_2)$ and 
$\interior(A_1) \subseteq \interior(A_2)$.

\begin{definition}[Quasi-discrete \cs{} - \qdcs{}]\label{def:QDClosureSpace}
A {\em quasi-discrete closure space} is a \cs{} $(X,\closure)$ such that any 
of the following equivalent conditions holds:
\begin{enumerate}
\item each $x \in X$ has a minimal neighbourhood;
\item for each $A \subseteq X$ it holds that $\closure(A) =
\bigcup_{x\in A}\closure(\SET{x})$.\closedefi
\end{enumerate}
\end{definition}

Given a relation $R \subseteq X \times X$, let function 
$\closure_{R}: \pws(X) \to \pws(X)$ be defined as follows: for all $A \subseteq X$,
$
\closure_R(A) = A \cup \ZET{x \in X}{\mbox{there exists } a \in A \mbox{ s.t. } (a,x)\in R}.
$
It is easy to see that, for any $R$, $\closure_{R}$ satisfies all the axioms of Definition~\ref{def:ClosureSpace} and so $(X, \closure_{R})$ is a \cs{.}
The following theorem is a standard result in the theory of \cs{s}~\cite{Gal03}:
\begin{theorem}
A \cs{} $(X, \closure)$ is quasi-discrete if and only if there is a relation $R \subseteq X \times X$ such that $\closure = \closure_{R}$. \qed
\end{theorem}

In the sequel, whenever $(X, \closure)$ is quasi-discrete, we will
let $\closureF$ denote $\closure_R$, and, consequently, we will
let $(X, \closureF)$ denote the space, abstracting from the specification of relation $R$, 
when the latter is not necessary. 
Moreover, we will let $\closureT$ denote $\closure_{R^{-1}}$.
$\interiorF$ and $\interiorT$ are defined in the obvious way: 
$\interiorF A= \overline{\closureF(\overline{A})}$ and
$\interiorT A= \overline{\closureT (\overline{A})}$. 

An example of  the difference between $\closureF$ and $\closureT$ is shown in Figure~\ref{fig:CloExample}.

\begin{figure}
\def\samplesz{2cm}
\centering
\subfloat[][]
{
\includegraphics[height=\samplesz]{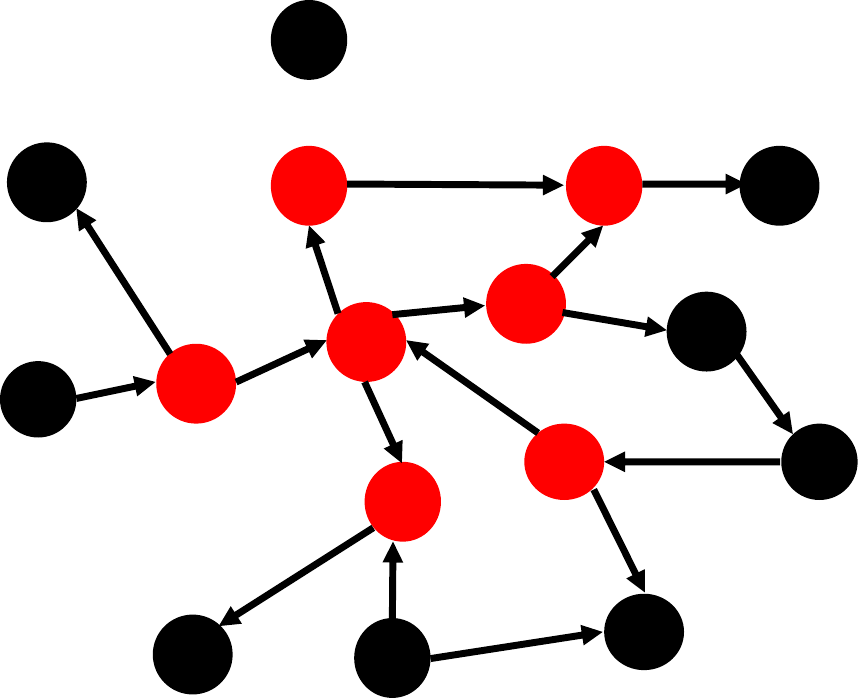}
\label{fig:Phi}
}\quad
\centering
\subfloat[][]
{
\includegraphics[height=\samplesz]{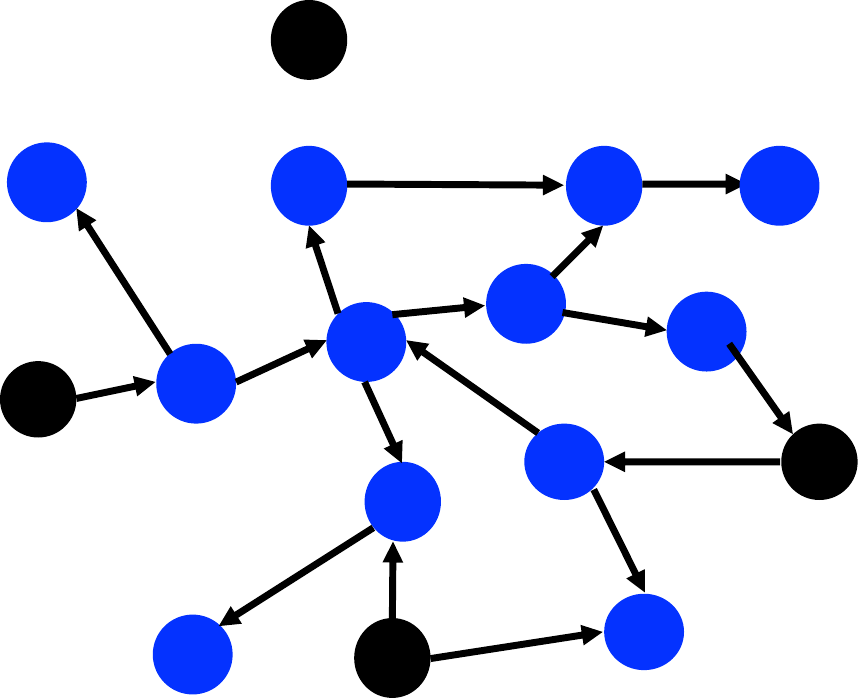}
\label{fig:NPhi}
}\quad
\centering
\subfloat[][]
{
\includegraphics[height=\samplesz]{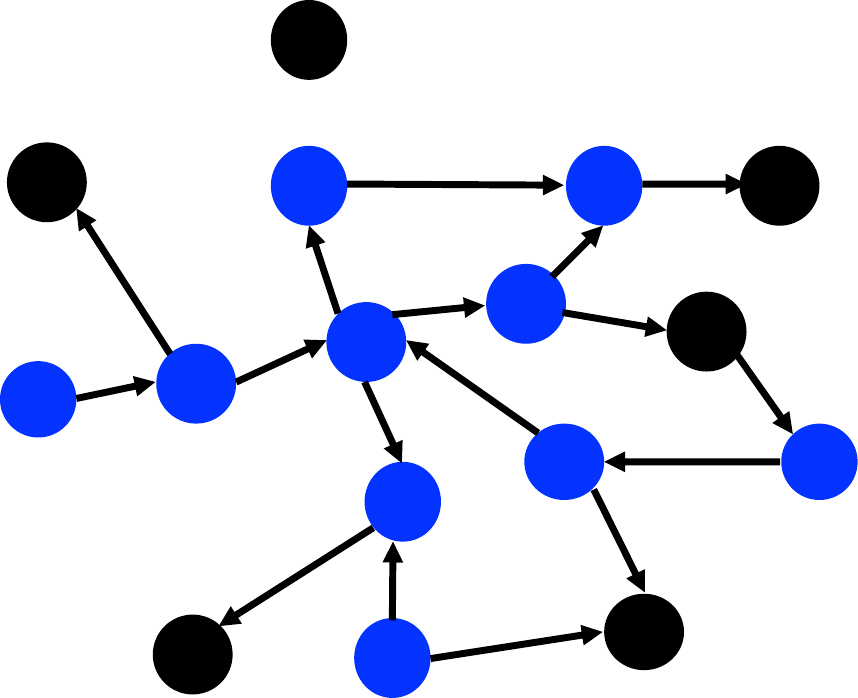}
\label{fig:RNPhi}
}
\caption{Given the points satisfying $\form$, shown in red in Fig.~\ref{fig:Phi}, those satisfying 
$\closureF(\form)$ are shown in blue in Fig.~\ref{fig:NPhi} and those satisfying 
$\closureT(\form)$ are shown in blue in Fig.~\ref{fig:RNPhi} }
\label{fig:CloExample}

\end{figure}

In the context of the present paper, {\em paths} over closure spaces play an important role;
therefore, we give a formal definition of paths as continuous functions below.

\begin{definition}[Continuous function]\label{def:ContinuousFunction}
Function $f : X_1 \to X_2$ is a {\em continuous function} from $(X_1,\closure_1)$ to $(X_2,\closure_2)$
if and only if for all sets $A \subseteq X_1$ we have: 
$f (\closure_1(A)) \subseteq \closure_2(f(A))$. \closedefi
\end{definition}

\begin{definition}[Connected space]\label{def:ConnectedSpace}
Given \cs{} $(X,\closure)$, $A \subseteq X$ is {\em connected} if
it is {\em not} the union of two non-empty separated sets.
$A_1, A_2 \subseteq X$ are {\em separated} if 
$A_1 \cap \closure(A_2) = \closure(A_1) \cap A_2 = \emptyset$.
$(X,\closure)$ is {\em connected} if $X$ is connected. \closedefi
\end{definition}

\begin{definition}[Index space]
An {\em index space} is a {\em connected} \cs{} $(I,\closure)$ 
equipped with a total order $\leq \subseteq I \times I$ with a bottom element $0$. 
We write $\iota_1 < \iota_2$ whenever $\iota_1 \leq \iota_2$ and $\iota_1 \not= \iota_2$.\closedefi
\end{definition}

\begin{definition}[Path]\label{def:path}
A {\em path} in \cs{} $(X,\closure)$ is a continuous function from 
an index  space $\calJ=(I,\closure^{\calJ})$  to $(X,\closure)$.
Path $\pi$ is {\em bounded} if there exists $\ell \in I$ s.t. 
$\pi(\iota)=\pi(\ell)$ for all $\iota$ such that $\ell \leq \iota$; 
we call $\ell$ the {\em length} of $\pi$, written $\pthlen(\pi)$. 

For bounded path $\pi$ we define the {\em domain} of $\pi$, $\dom(\pi)$,  as the set
$\ZET{\iota}{\iota \leq \pthlen(\pi)}$ and 
$\rng(\pi) = \ZET{\pi(\iota)}{\iota \leq (\pthlen \,\pi)}$ (the {\em range} of $\pi$).
\closedefi
\end{definition}

Of particular importance in the present paper are {\em quasi-discrete} paths and {\em Euclidean}
paths. Quasi-discrete paths are paths having $(\nats,\closure_{\succ})$ as 
index space, where $\nats$ is the set of natural numbers and $\succ$ is the {\em successor} relation $\succ = \ZET{(m,n)}{n=m+1}$. The index space of Euclidean paths is instead the set
of non-negative real numbers equipped with the classical closure operator.

\begin{proposition}\label{prp:FT}
For all \qdcs{} $(X, \closureF), A, A_1, A_2 \subseteq X, x_1,x_2 \in X$, and function $\pi: \nats \to X$ the following holds:
\begin{enumerate}
\item $\closureT(A) = 
A \cup \ZET{x \in X}{\mbox{there exists } a \in A \mbox{ such that } (x,a)\in R}$;
\item
$x_1 \in \closureT(\SET{x_2})$ if and only if $x_2 \in \closureF(\SET{x_1})$;
\item 
$\closureT(A)=\ZET{x}{x\in X \mbox{ and exists } a \in A \mbox{ such that } a \in \closureF(\SET{x})}$;
\item if $A_1 \subseteq A_2$, then $\closureT(A_1) \subseteq \,\closureT(A_2)$ and
$\interiorT(A_1) \subseteq \,\interiorT(A_2)$.
\item\label{path}  $\pi$ is a path over $X$ if and only if for all  
$i \in \dom(\pi) \setminus \SET{0}$, the following holds:\\
$ \pi(i) \in \closureF(\pi(i-1))$  and  $\pi(i-1) \in \closureT(\pi(i))$.
\end{enumerate}
\end{proposition}

\begin{remark}
Note that the {\em definition} of the closure operator for \qdcs{s} given in~\cite{Gal03} coincides with  $\closureT A$, as given in the first item of Proposition~\ref{prp:FT}.\closerem
\end{remark}

In the sequel we fix a set  $\ap$  of {\em atomic proposition letters}.

\begin{definition}[Closure model - \cm{}]\label{def:ClosureModel}
A {\em closure model}, \cm{} for short, is a tuple $\model=(X,\closure, \peval)$, with $(X,\closure)$ a \cs{,} and $\peval: \ap \to \pws(X)$ the (atomic predicate) valuation function assigning to each $p\in \ap$ the set of points where $p$ holds. \closedefi
\end{definition}

The following definition adapts  the  notion of homeomorphism for topological spaces, as given in~\cite{Mor21}, to the case of closure spaces.

\begin{definition}[Homeomorphism]\label{def:Homeo}
A homeomorphism between \cm{s} $\model_1 = (X_1,\closure_1, \peval_1)$ and
$\model_2 = (X_2,\closure_2, \peval_2)$ is
a {\em bijection} $h:X_1 \to X_2$ s.t. for all $x_1 \in X_1$, $x_2 \in X_2, A_1 \subseteq X_1$ and $A_2 \subseteq X_2$, the following holds:
\begin{enumerate}
\item $\invpeval_1(x_1) = \invpeval_2(h(x_1))$;
\item $\invpeval_2(x_2) = \invpeval_1(h^{-1}(x_2))$;
\item $h(\interior_1(A_1)) = \interior_2(h(A_1))$;
\item $h^{-1}(\interior_2(A_2)) = \interior_1(h^{-1}(A_2))$.
\end{enumerate}
We say that $x_1,x_2 \in X$ are homeomorphic, written $x_1\homeo x_2$ if and only if there
is an  homeomorphism $h$ such that $x_2=h(x_1)$.
\closedefi
\end{definition}

\noindent
An alternative, equivalent, definition can be obtained by requiring that
$h(\closure_1(A_1))$ $=$ $\closure_2(h(A_1))$ and 
$h^{-1}(\closure_2(A_2)) = \closure_1(h^{-1}(A_2))$
instead of 
$h(\interior_1(A_1)) = \interior_2(h(A_1))$ and
$h^{-1}(\interior_2(A_2)) = \interior_1(h^{-1}(A_2))$.
%

All the definitions given above for \cs{s} apply to \cm{s} as well; thus, a
{\em quasi-discrete closure model} (\qdcm{} for short) is a \cm{} $\model=(X, \closureF,\peval)$ where $(X,\closureF)$ is a \qdcs{.}
For model  $\model=(X,\closure,\peval)$ we will often write $x\in \model$ when $x \in X$; similarly we will speak of paths in $\model$ meaning paths in $(X,\closure)$;
we let 
$\pths_{\calJ,\model}$ denote the set of all paths in $\model$ with index space
$\calJ$. 
$\bpths_{\calJ,\model}$ denotes the set of all {\em bounded} paths in $\model$, whereas
for $x \in X$, $\bpthsF_{\calJ,\model}(x)$ denotes the set 
$\ZET{\pi\in \bpths_{\calJ,\model}}{\pi(0)=x}$ and, similarly,
$\bpthsT_{\calJ,\model}(x)$ denotes the set 
$\ZET{\pi\in \bpths_{\calJ,\model}}{\pi(\pthlen(\pi))=x}$. We will refrain from writing the
subscripts $_{\calJ,\model}$ whenever not necessary.

We often write $x\leadspth{\pi} x'$ if  $\pi \in \bpthsF_{\calJ,\model}(x) \cap \bpthsT_{\calJ,\model}(x')$ and
$x\leadspth{} x'$ if there exists $\pi$ s.t. $x\leadspth{\pi} x'$.
We say that $\model$ is {\em path-connected} if for all $x,x' \in \model$ we have $x\leadspth{} x'$.

Finally, for $\pi \in \pths_{\calJ,\model}$ with $\calJ=(I,\closure^{\calJ})$, 
we let $\T(\pi)$ denote the {\em trace} of $\pi$, namely 
$\T: \pths_{\calJ,\model} \to (I \to \pws(\ap))$ with
$\T(\pi)(\iota) = \invpeval(\pi(\iota))$. We say that $x_1,x_2\in \model$ are {\em trace equivalent}, written $x_1 \, \treq \, x_2$ if $\T(\bpthsF(x_1)) = \T(\bpthsF(x_2))$ and
$\T(\bpthsT(x_1)) = \T(\bpthsT(x_2))$.

In the sequel, for logic $\calL$, formula $\form \in \calL$, and model $\model=(X,\closure,\peval)$ we let  $\sem{\form}^{\model}_{\calL}$ denote the set $\ZET{x\in X}{\model,x \models_{\calL} \form}$ of all the points in $\model$ that satisfy $\form$, where $\models_{\calL}$ is the satisfaction relation for $\calL$. For the sake of readability, we will refrain from writing the subscript $_{\calL}$ when this will not cause confusion.

\section{\cm-bisimilarity}\label{sec:CMbisimilarity}

\subsection{\cm-bisimilarity}

The first notion of bisimilarity that we consider is \cm-bisimilarity. This notion stems from a natural adaptation for \cm{s} of Topo-bisimilation for topological models, as defined e.g. in~\cite{vBB07}. We recall such definition below, where $(X,\tau,\peval)$ is the topological model with set of points $X$, open sets $\tau$,  and atomic predicate evaluation function $\peval$:

\begin{definition}[Topo-bisimulation]\label{def:Topobisimulation}
A {\em topological bisimulation} or simply a {\em topo-bisimulation} between two topo-models 
$\model_1=(X_1,\tau_1,\peval_1)$ and $\model_2= (X_2,\tau_2,\\\peval_2)$
is a non-empty relation $T\subseteq X_1 \times X_2$ such that if $(x_1,x_2)\in T$ then:
\begin{enumerate}
\item $x_1 \in \peval_1(p)$ if and only if $x_2 \in \peval_2(p)$ for each $p \in \ap$;
\item (forth): $x_1 \in U_1 \in \tau_1$ implies there exists $U_2 \in\tau_2$ such that $x_2 \in U_2$ and for all $x_2' \in U_2$ there exists $x_1' \in U_1$ such that $(x_1',x_2')\in T $;
\item (back): $x_2 \in U_2 \in \tau_2$ implies there exists $U_1 \in\tau_1$ such that $x_1 \in U_1$ and for all $x_1' \in U_1$ there exists $x_2' \in U_2$ such that $(x_1',x_2')\in T $. \closedefi
\end{enumerate}
\end{definition}

In the context of \cm{s}, we replace the notion of {\em open set} containing a given point $x$ with that of {\em neighbourhood} of $x$, so that we get the following\footnote{In this paper, we provide all major definitions and result with respect to a {\em single} model
$\model=(X,\closure,\peval)$ whereas some authors do this with respect to {\em two} models $\model_1=(X_1,\closure_1,\peval_1)$ and $\model_2=(X_2,\closure_2,\peval_2)$. The two approaches are interchangeable and we find the former a little bit simpler from the notational point of view.}

\begin{definition}[\cm-bisimilarity]\label{def:CMbisimilarity}
Given CM $\model=(X,\closure, \peval)$, a non empty relation 
$B \subseteq X \times X$ is a {\em \cm-bisimulation  over $X$} if, whenever $(x_1,x_2) \in B$, the following holds:
\begin{enumerate}
\item $\invpeval(x_1) = \invpeval(x_2)$;
\item for all neighbourhoods $S_1$ of $x_1$  there is a neighbourhood $S_2$ of $x_2$ such that
for all $s_2 \in S_2$, there is $s_1 \in S_1$ with $(s_1,s_2)\in B$;
\item for all neighbourhoods $S_2$ of $x_2$ there is a neighbourhood $S_1$ of $x_1$ such that
for all $s_1 \in S_1$, there is $s_2 \in S_2$ with $(s_1,s_2)\in B$. 
\end{enumerate}
$x_1$ and $x_2$ are {\em \cm-bisimilar}, written $x_1\,\cmbis^{\model}\, x_2$,  if and only if there is a \cm-bisimulation $B$ over $X$  such that $(x_1,x_2) \in B$. \closedefi
\end{definition}

The above definition is very similar to that of bisimilarity between {\em monotonic neighbourhood spaces}~\cite{vB+17,Han03}, and, in fact, monotonicity of the $\interior$ operator makes it legitimate to interpret \cm{s} as an instantiation of the notion of monotonic neighbourhood models (see~\cite{vB+17,Han03} for details).\\

\cm-bisimilarity is coarser than homeomorphism:

\begin{proposition}
\label{prp:HomeoImplCMbis}
For all \cm{s} $\model = (X,\closure, \peval)$ and $x_1, x_2 \in X$ the following holds:
$x_1\, \homeo \, x_2$ implies $x_1\, \cmbis \, x_2$.
\end{proposition}

The converse of \propo{}~\ref{prp:HomeoImplCMbis} does not hold
as shown in Figure~\ref{fig:CMbisNoImplHomeo} where 
$
\invpeval(x_{11})=\invpeval(x_{21})=\SET{r}\not=\SET{b}=\invpeval(x_{12})=\invpeval(x_{22})=\invpeval(x_{23})
$
and $x_{11}  \cmbis x_{21}$ but $x_{11}  \not\homeo x_{21}$ 
(see Remark~\ref{rem:prp:HomeoImplCMbis} in Appendix~\ref{apx:sec:CMbisimilarity}).
\begin{figure}
\centerline{\resizebox{1.5in}{!}{\includegraphics{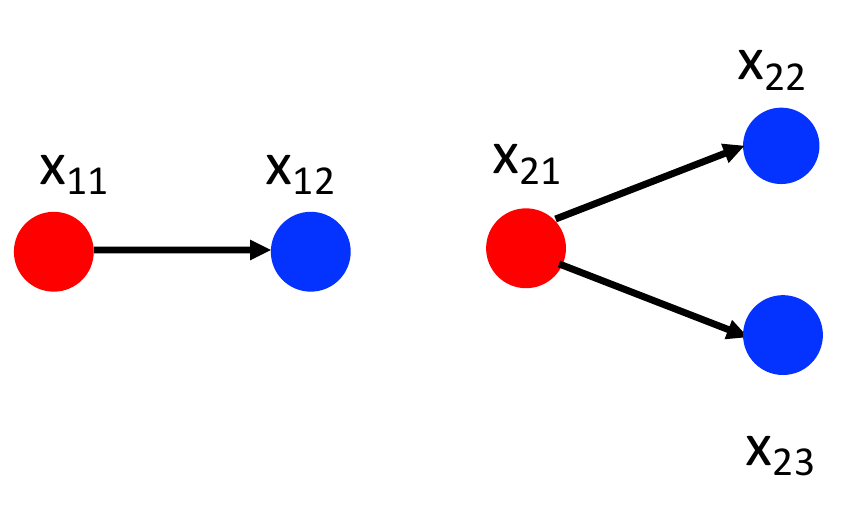}}}
\caption{$x_{11}  \cmbis x_{21}$ but $x_{11}  \not\homeo x_{21}$.\label{fig:CMbisNoImplHomeo}}
\end{figure}

\subsection{Logical Characterisation of \cm-bisimilarity}

In this section, we show that \cm-bisimilarity is characterised by the Infinitary Modal Logic, \iml{} for short. We first recall the definition  of \iml{.}

\begin{definition}[Infinitary Modal Logic - \iml]\label{def:Iml}
For index set $I$ and $p \in  \ap$ the abstract language of \iml{} is defined as follows:
$$
\form ::= p \; \sep \; \lneg \form \; \sep \; \liand_{i \in I} \form_i \; \sep \; \lnear \form.
$$
The satisfaction relation for all \cm{s}  $\model$, points $x \in \model$, and \iml{} formulas $\form$ is defined recursively on the structure of $\form$ as follows:
$$
\begin{array}{r c l c l c l l}
\model,x & \models_{\iml}  & p & \Leftrightarrow & x  \in \peval(p);\\
\model,x & \models_{\iml}  & \lneg \,\form & \Leftrightarrow & \model,x  \models_{\iml} \form \mbox{ does not hold};\\
\model,x & \models_{\iml}  & \liand_{i\in I} \form_i  & \Leftrightarrow &
\model,x  \models_{\iml} \form_i \mbox{ for all } i \in I;\\
\model,x & \models_{\iml}  & \lnear \form & \Leftrightarrow & x \in \closure(\sem{\form}^{\model}).
\end{array}
$$

\closedefi
\end{definition}


\begin{definition}[\iml-Equivalence]\label{def:ImlEq}
Given \cm{} $\model = (X,\closure,\peval)$, the equivalence relation $\imleq^{\model} \, \subseteq \, X \times X$ is defined as: $x_1 \imleq^{\model} x_2$ if and only if for all \iml{} formulas $\form$ the following holds: $\model, x_1 \models_{\iml} \form$ if and only if $\model, x_2 \models_{\iml} \form$.
\closedefi
\end{definition}

In the sequel we will often abbreviate $\imleq^{\model}$ with $\imleq$, leaving the specification of the model implicit.

\begin{theorem}\label{thm:CMbisIsImleq}
For all \cm{s} $\model=(X,\closure,\peval)$, any
\cm-Bisimulation   $B$ over $X$ is included in the equivalence $\imleq^{\model}$.
\closethm
\end{theorem}

The  converse of Theorem~\ref{thm:CMbisIsImleq} is given below.

\begin{theorem}\label{thm:ImleqIsCMbis}
For all \cm{s}  $\model=(X,\closure,\peval)$, $\imleq^{\model}$ is a \cm-Bisimulation.
\closethm
\end{theorem}

\begin{corollary}\label{cor:CMbisEqImleq}
For all \cm{s}  $\model=(X,\closure,\peval)$ we have that $\imleq^{\model}$ coincides with $\cmbis^{\model}$.
\qed 
\end{corollary}

\section{\cmc-bisimilarity for Quasi-discrete CMs}\label{sec:CMCbisimilarity}
In this section we refine \cm-bisimilarity into {\em \cm-bisimilarity with converse}, \cmc-bisimilarity for short, a specialisation of \cm-bisimilarity for \qdcm{s}. Recall that, for \cm{} 
$\model = (X,\closure,\peval)$, $S \subseteq X$ is a neighbourhood of $x \in X$ if
$x \in \interior(S)$. Moreover, whenever $\model$ is quasi-discrete, there are actually two
interior functions, namely $\interiorF(S)$ and $\interiorT(S)$.  It is then natural to exploit both 
functions for a definition of \cm-bisimilarity specifically designed for \qdcm{s}, namely \cmc-bisimilarity.

\subsection{\cmc-bisimilarity for \qdcm{s}}

\begin{definition}[\cmc-bisimilarity for \qdcm{s}]\label{def:CMCbisimilarity}
Given \qdcm{} $\model=(X,\closureF, \\\peval)$, a non empty relation 
$B \subseteq X \times X$ is a {\em \cmc-bisimulation  over $X$} if, whenever $(x_1,x_2) \in B$, the following holds:
\begin{enumerate}
\item $\invpeval(x_1) = \invpeval(x_2)$;
\item for all $S_1 \subseteq X$ such that  $x_1 \in \interiorF(S_1)$  there is $S_2 \subseteq X$ such that $x_2 \in \interiorF(S_2)$ and 
for all $s_2 \in S_2$, there is $s_1 \in S_1$ with $(s_1,s_2)\in B$;
\item for all $S_2 \subseteq X$ such that  $x_2 \in \interiorF(S_2)$  there is $S_1 \subseteq X$ such that $x_1 \in \interiorF(S_1)$ and 
for all $s_1 \in S_1$, there is $s_2 \in S_2$ with $(s_1,s_2)\in B$; 
\item for all $S_1 \subseteq X$ such that  $x_1 \in \interiorT(S_1)$  there is $S_2 \subseteq X$ such that $x_2 \in \interiorT(S_2)$ and 
for all $s_2 \in S_2$, there is $s_1 \in S_1$ with $(s_1,s_2)\in B$;
\item for all $S_2 \subseteq X$ such that  $x_2 \in \interiorT(S_2)$  there is $S_1 \subseteq X$ such that $x_1 \in \interiorT(S_1)$ and 
for all $s_1 \in S_1$, there is $s_2 \in S_2$ with $(s_1,s_2)\in B$.
\end{enumerate}
$x_1$ and $x_2$ are {\em CMC-bisimilar}, written $x_1\,\cmcbis^{\model}\, x_2$,  if and only if there is a \cmc-bisimulation $B$ over $X$  such that $(x_1,x_2) \in B$. \closedefi
\end{definition}

The following proposition trivially follows from the relevant definitions, keeping in mind that, for \qdcm{s} $\interior$ coincides with  $\interiorF$.

\begin{proposition}
\label{prp:CMCbisImpliesCMbis}
For all \qdcm{s} $\model = (X,\closureF, \peval)$ and $x_1, x_2 \in X$ the following holds:
$x_1\, \cmcbis \, x_2$ implies $x_1\,  \cmbis \, x_2$.\qed
\end{proposition}

The converse of \propo{}~\ref{prp:CMCbisImpliesCMbis} does not hold
as shown in Figure~\ref{fig:CMbisNoImplCMCbis} 
\begin{figure}
\centerline{\resizebox{1.2in}{!}{\includegraphics{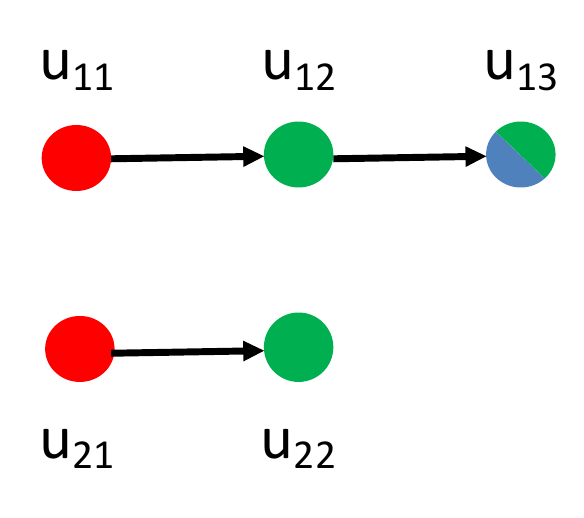}}}
\caption{$u_{11}  \cmbis u_{21}$ but $u_{11}  \not\cmcbis u_{21}$.\label{fig:CMbisNoImplCMCbis}}
\end{figure}
where $\invpeval(u_{11})=\invpeval(u_{21})=\SET{r}, \invpeval(u_{12})=\invpeval(u_{22})=\SET{g}$ and $\invpeval(u_{13})=\SET{b,g}$ (see Remark~\ref{rem:prp:CMCbisImpliesCMbis} in Appendix~\ref{apx:sec:CMCbisimilarity}).

\begin{proposition}
\label{prp:CMCbisImpliesTraceEq}
For all \qdcm{s} $\model = (X,\closureF, \peval)$ and $x_1, x_2 \in X$ the following holds:
$x_1\, \cmcbis \, x_2$ implies $x_1\,  \treq \, x_2$.
\end{proposition}

The converse of \propo{}~\ref{prp:CMCbisImpliesTraceEq} does not hold
as shown in Figure~\ref{fig:TraceEqNoImplCMCbis} where 
$
\invpeval(y_{11})=\invpeval(y_{12})=\invpeval(y_{21})=\invpeval(y_{22})=\invpeval(y_{24})=\SET{r}\not=\SET{b}=\invpeval(y_{13})=\invpeval(y_{23})
$
and $y_{11}  \treq y_{21}$ but $y_{11}  \not\cmcbis y_{21}$ 
(see  Remark~\ref{rem:prp:CMCbisImpliesTraceEq} in Appendix~\ref{apx:sec:CMCbisimilarity}).
\begin{figure}
\centerline{\resizebox{1.2in}{!}{\includegraphics{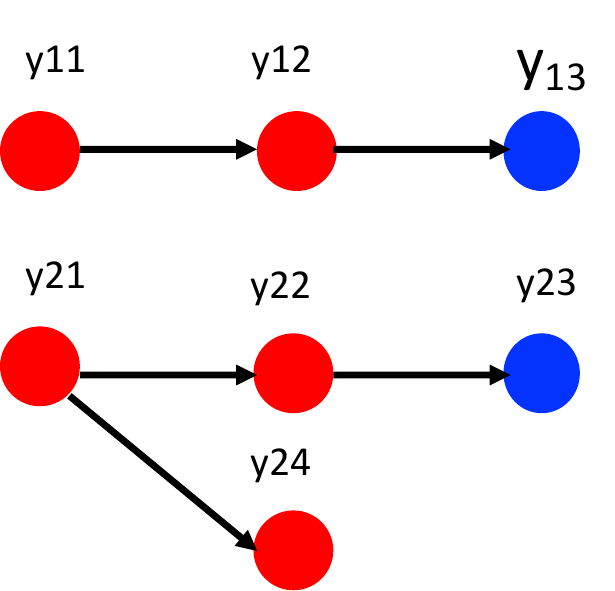}}}
\caption{$y_{11}  \treq y_{21}$ but $y_{11}  \not\cmcbis y_{21}$.\label{fig:TraceEqNoImplCMCbis}}
\end{figure}

\subsection{Logical Characterisation of \cmc-bisimilarity for \qdcm{s}}
In order to provide a logical characterisation of \cmc-bisimilarity, we extend \iml{} with a
``converse'' of the modal operator of classical \iml{,} thus exploiting the inverse 
of the binary relation underlying the \qdcm{.} The result is a logic with the two modalities
$\lnearF$ and $\lnearT$, with the expected meaning.

\begin{definition}[Infinitary Modal Logic with Converse - \imlc{}]\label{def:Imlc}
For index set $I$ and $p \in  \ap$ the abstract language of \imlc{} is defined as follows:
$$
\form ::= p \; \sep \; \lneg \form \; \sep \; \liand_{i \in I} \form_i \; \sep \; 
\lnearF \form \; \sep \; \lnearT \form.
$$
The satisfaction relation for all \qdcm{s}  $\model$, points $x \in \model$, and \imlc{} formulas $\form$ is defined recursively on the structure of $\form$ as follows:
$$
\begin{array}{r c l c l c l l}
\model,x & \models_{\imlc}  & p & \Leftrightarrow & x  \in \peval(p);\\
\model,x & \models_{\imlc}  & \lneg \,\form & \Leftrightarrow & \model,x  \models_{\imlc} \form \mbox{ does not hold};\\
\model,x & \models_{\imlc}  & \liand_{i\in I} \form_i  & \Leftrightarrow &
\model,x  \models_{\imlc} \form_i \mbox{ for all } i \in I;\\
\model,x & \models_{\imlc}  & \lnearF \form & \Leftrightarrow & x \in \closureF(\sem{\form}^{\model});\\
\model,x & \models_{\imlc}  & \lnearT \form & \Leftrightarrow & x \in \closureT(\sem{\form}^{\model}).
\end{array}    
$$

\closedefi
\end{definition}

\begin{definition}[\imlc-Equivalence]\label{def:ImlcEq}
Given \qdcm{} $\model = (X,\closureF,\peval)$, the equivalence relation $\imlceq^{\model} \, \subseteq \, X \times X$ is defined as: $x_1 \imlceq^{\model} x_2$ if and only if for all \imlc{} formulas $\form$ the following holds: $\model, x_1 \models_{\imlc} \form$ if and only if $\model, x_2 \models_{\imlc} \form$.
\closedefi
\end{definition}

In the sequel we will often abbreviate $\imlceq^{\model}$ with $\imlceq$.

\begin{theorem}\label{thm:CMCbisIsImlceq}
For all \qdcm{s} $\model=(X,\closureF,\peval)$, any
\cmc-Bisimulation   $B$ over $X$ is included in the equivalence $\imlceq^{\model}$.
\closethm
\end{theorem}

The  converse of Theorem~\ref{thm:CMCbisIsImlceq} is given below.

\begin{theorem}\label{thm:ImlceqIsCMCbis}
For all \qdcm{s}  $\model=(X,\closureF,\peval)$, $\imlceq^{\model}$ is a \cmc-Bisimulation.
\closethm
\end{theorem}

\begin{corollary}\label{cor:CMCbisEqImleq}
For all \qdcm{s}  $\model=(X,\closureF,\peval)$ we have that $\imlceq^{\model}$ coincides with $\cmcbis^{\model}$.
\qed 
\end{corollary}

\subsection{\cl-bisimilarity for \qdcm{s}}
In this section, we recall a notion of bisimilarity for \qdcm{s} that has been
proposed in~\cite{Ci+20} and that is based on closure functions, instead of interior functions. 
We then prove that such a notion, which here we call \cl-bisimilarity, 
coincides with \cmc-bisimilarity. 
The introduction of \cl-bisimilarity  is motivated by the fact that we find it more intuitive, and  its use makes several proofs simpler.

\begin{definition}[\cl-bisimilarity for \qdcm{s}]\label{def:Clbisimilarity}
Given \qdcm{} $\model=(X,\closureF, \peval)$, a non empty relation 
$B \subseteq X \times X$ is a {\em \cl-bisimulation  over $X$} if, whenever $(x_1,x_2) \in B$, the following holds:
\begin{enumerate}
\item $\invpeval{x_1} = \invpeval{x_2}$;
\item for all $x_1' \in  \closureF(\SET{x_1})$ 
there exists $x_2' \in  \closureF(\SET{x_2})$ such that $(x_1',x_2') \in B$;
\item for all $x_2' \in  \closureF(\SET{x_2})$ 
there exists $x_1' \in  \closureF(\SET{x_1})$ such that $(x_1',x_2') \in B$;
\item for all $x_1' \in  \closureT(\SET{x_1})$ 
there exists $x_2' \in  \closureT(\SET{x_2})$ such that $(x_1',x_2') \in B$;
\item for all $x_2' \in  \closureT(\SET{x_2})$ 
there exists $x_1' \in  \closureT(\SET{x_1})$ such that $(x_1',x_2') \in B$;
\end{enumerate}
We say that $x_1$ and $x_2$ are \cl-bisimilar, written  $x_1 \, \clbis^{\model} \, x_2$, if and only if there exists a \cl-bisimulation $B$ such that $(x_1, x_2)\in B$.
\closedefi
\end{definition}

As mentioned in Section~\ref{sec:Introduction}, \cl-bisimulation resembles (strong) Back and Forth bisimulation of~\cite{De+90}, in particular for the presence of Conditions 4 and 5.
Should we delete the above mentioned conditions, thus making our definition of \cl-bisimulation more similar to classical strong bisimulation for transition systems, we would have to consider points $v_{12}$ and $v_{22}$ of Figure~\ref{fig:V12NOEQV22} bisimilar where
$\invpeval(v_{11})=\SET{r}\not=\SET{g}=\invpeval(v_{21})$ and 
$\invpeval(v_{12})=\SET{b}=\invpeval(v_{22})$.
\begin{figure}
\centerline{\resizebox{0.8in}{!}{\includegraphics{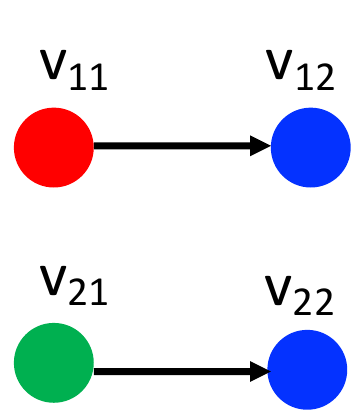}}}
\caption{$v_{12}$ and $v_{22}$ are not \cl-bisimilar.\label{fig:V12NOEQV22}}
\end{figure}
We instead want to consider them as not being bisimilar because they are in the closure of (i.e. they are ``near'' to) points that are not bisimilar, namely $v_{11}$ and $v_{21}$. For instance, $v_{21}$ might represent a poisoned physical location (whereas $v_{11}$ is not poisoned) and so $v_{22}$ should not be considered equivalent to $v_{12}$ because the former can be reached from the poisoned location while the latter cannot.

\subsection{\cl-bisimilarity minimisation}
In~\cite{Ci+20} we have shown a minimisation algorithm  for $\clbis^{\model}$.
The algorithm is defined in a coalgebraic setting: it takes an $\calF$-coalgebra, for appropriate functor $\calF$ in the
category {\bf Set}, and returns the bisimilarity quotient of its carrier set. The instantiation of the algorithm for (a
coalgebraic interpretation of) \qdcs{s} is implemented in the tool  MiniLogicA, available for the major operating systems at {\tt https://github.com/vincenzoml/MiniLogicA}.

\subsection{Logical Characterisation of \cl-bisimilarity}
In this section we present \islcs, an infinitary version of a variant of the {\em Spatial Logic for Closure Spaces} (\slcs). \slcs{} has been proposed in~\cite{Ci+16} and its basic modal operators are {\em near} ($\lnear$), {\em surrounded} ($\lsurr$) and {\em propagation} ($\lprop$), whereas reachability operators are  derived from the above. 
The variant of \slcs{} that we use in this section, instead, has only two basic reachability operators $\ltothru$ and $\lfromthru$, as in~\cite{Ci+20}. We will show that, when the underlying interptretation model is quasi-discrete, $\lnear$ can be derived from the reachability operators: more precisely, $\lnearF$ can be derived from $\lfromthru$ and $\lnearT$ from $\ltothru$ (Lemma~\ref{lem:NearEqRhoFalse} below)\footnote{For completeness, in Proposition~\ref{prp:SurrAndPropDerived} in Appendix~\ref{apx:prp:SurrAndPropDerived}, we also show  that, for {\em general} $\cm{s}$, 
$\lsurr$ can be derived from $\ltothru$ and $\lprop$ from $\lfromthru$, when the latter are interpreted over general \cm{s}.}.
Furthermore, we show that \islcs{} characterises \cl-bisimilarity. We also prove that an appropirate sub-logic of \islcs{} is sufficient for characterising \cl-bisimilarity and that
such sub-logic coincides with \imlc. As a side-result, we get the coincidence of 
\cmc-bisimilarity and \cl-bisimilarity.

\begin{definition}[Infinitary \slcs{} - \islcs{}]\label{def:Islcs}
For index set $I$ and $p \in  \ap$ the abstract language of \islcs{} is defined as follows:
$$
\form ::= p \; \sep \; \lneg \form \; \sep \; \liand_{i \in I} \form_i \; \sep \; 
\ltothru \form_1 [\form_2] \; \sep \; \lfromthru \form_1  [\form_2].
$$
The satisfaction relation for \qdcm{s}  $\model$, $x \in \model$, and \islcs{} formulas $\form$ is defined recursively on the structure of $\form$ as follows:
$$
\begin{array}{r c l c l c l}
\model,x & \models_{\islcs}  & p & \Leftrightarrow & x  \in \peval(p)\\
\model,x & \models_{\islcs}  & \lneg \,\form & \Leftrightarrow & \model,x  \models_{\islcs} \form \mbox{ does not hold}\\
\model,x & \models_{\islcs}  & \liand_{i\in I} \form_i  & \Leftrightarrow &
\model,x  \models_{\islcs} \form_i \mbox{ for all } i \in I\\
\model,x  & \models_{\islcs} & \ltothru \form_1 [\form_2] & \Leftrightarrow &
\mbox{there exist  path } \pi \mbox{ and index } \ell \mbox{ such that }\\
&&&&\mbox{\hspace{0.1in}} \pi(0) = x \mbox{ and }\\
&&&&\mbox{\hspace{0.1in}} \pi(\ell) \models_{\islcs} \form_1 \mbox{ and }\\
&&&&\mbox{\hspace{0.1in}} \mbox{for all } j \mbox{ such that } 0 < j < \ell \mbox{ the following holds:}\\
&&&&\mbox{\hspace{0.2in}} \pi(j) \models_{\islcs} \form_2;\\
\model,x  & \models_{\islcs} & \lfromthru \form_1 [\form_2] & \Leftrightarrow &
\mbox{there exist  path } \pi \mbox{ and index } \ell \mbox{ such that }\\
&&&&\mbox{\hspace{0.1in}} \pi(\ell) = x \mbox{ and }\\
&&&&\mbox{\hspace{0.1in}} \pi(0) \models_{\islcs} \form_1 \mbox{ and }\\
&&&&\mbox{\hspace{0.1in}} \mbox{for all } j \mbox{ such that } 0 < j < \ell \mbox{ the following holds:}\\
&&&&\mbox{\hspace{0.2in}} \pi(j) \models_{\islcs} \form_2.
\end{array}
$$
\closedefi
\end{definition}

\begin{definition}[\islcs-Equivalence]\label{def:IslcslEq}
Given \qdcm{} $\model = (X,\closureF,\peval)$, the equivalence relation $\islcseq^{\model} \, \subseteq \, X \times X$ is defined as: $x_1 \islcseq^{\model} x_2$ if and only if for all \islcs{} formulas $\form$ the following holds: $\model, x_1 \models_{\islcs} \form$ if and only if $\model, x_2 \models_{\islcs} \form$.
\closedefi
\end{definition}

\begin{theorem}\label{thm:ClbisIsIslcseq}
For all \qdcm{s} $\model=(X,\closureF,\peval)$, any
\cl-Bisimulation   $B$ over $X$ is included in the equivalence $\islcseq^{\model}$.
\closethm
\end{theorem}

The  converse of Theorem~\ref{thm:ClbisIsIslcseq} is given below.

\begin{theorem}\label{thm:IslcseqIsClbis}
For all \qdcm{s}  $\model=(X,\closureF,\peval)$, $\islcseq^{\model}$ is a \cl-Bisimulation.
\closethm
\end{theorem}

\begin{corollary}\label{cor:ClbisEqIslcseq}
For all \qdcm{s}  $\model=(X,\closureF,\peval)$ we have that $\islcseq^{\model}$ coincides with $\clbis^{\model}$.
\qed 
\end{corollary}
 
Let $\lnearF$ and $\lnearT$ be defined as in Definition~\ref{def:Imlc}.

\begin{lemma}\label{lem:NearEqRhoFalse}
For all \qdcm{s}  $\model=(X,\closureF,\peval)$, the following holds:\\
$\lnearF \form \equiv \lfromthru \form[\lfalse]$ and $\lnearT \form \equiv \ltothru \form [\lfalse]$.
\end{lemma}

\begin{theorem}\label{thm:ClbisEqImlceq}
For all \qdcm{s}  $\model=(X,\closureF,\peval)$, $\imlceq^{\model}$ coincides with $\clbis^{\model}$.
\closethm
\end{theorem}

From Corollary~\ref{cor:CMCbisEqImleq} and Theorem~\ref{thm:ClbisEqImlceq}
we get the following

\begin{corollary}\label{cor:ClbisEqCMCbis}
For all \qdcm{s}  $\model=(X,\closureF,\peval)$ we have that $\cmcbis^{\model}$ coincides with $\clbis^{\model}$.
\qed 
\end{corollary}

\section{\pth-bisimilarity}\label{sec:Pathbisimilarity}

\cm-bisimilarity, and its refinements \cmc-bisimilarity and \cl-bisimilarity, are a fundamental starting point for the study of bisimulations in space due to their strong links to Topo-bisimulation. On the other hand, they are somehow too much fine grain relations for reasoning about general properties of space and related notions of model minimisation. For instance, with reference to the model of Figure~\ref{fig:FiveAndFour}, where all red points satisfy only atomic proposition $r$ while the blue ones satisfy only $b$, the point at the center of the left part of the model is not \cmc-bisimilar to any other red point in the model. This is because \cmc-bisimilarity is based on the fact that points reachable ``in one step'' are taken into consideration, as it is clear from the equivalent \cl-bisimilarity definition. This, in turn,  gives bisimilarity a sort of ``counting'' power, that goes against the idea that, for instance, the left  part of the model could be represented by the right part---and that, actually, both parts could be represented by a minimal model consisting of just one red point and one blue point, connected by a symmetric arrow, which would convey an idea of space scaling. Such scaling would be quite useful when dealing, for instance, with models representing images---as briefly mentioned in Section~\ref{sec:Introduction}. Such models are \qdcm{s} where the ``points'' are pixels or voxels and the underlying relation is the so called {\em Adjacency} relation, i.e. a reflexive and symmetric relation such that each pixel/voxel is related to all the pixel/voxel that share an edge or a vertex with it. In this and in the next sections, we present  weaker notions of bisimilarity, namely \pth-bisimilarity and \cop-bisimilarity, with the aim of capturing the intuitive notions briefly discussed above. We start with the definition of \pth-bisimilarity.

\begin{figure}
\centerline{
\resizebox{1.3in}{!}{\includegraphics{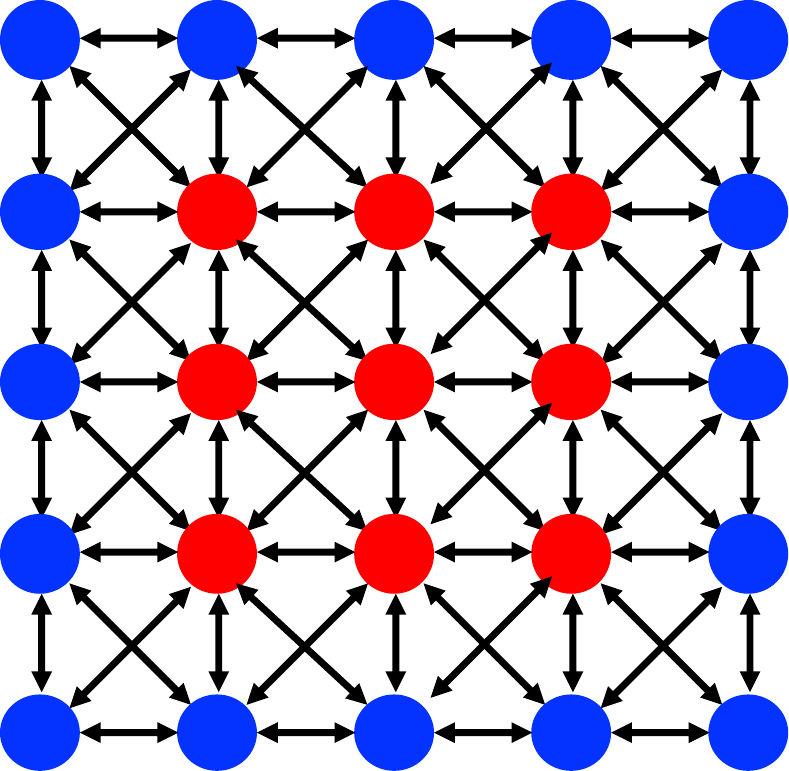}}\hspace{0.7in}
\raisebox{0.1in}{\resizebox{1in}{!}{\includegraphics{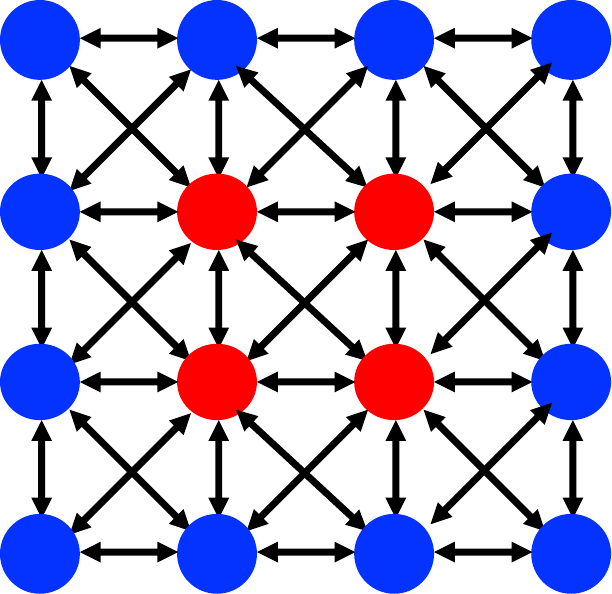}}}
}
\caption{\label{fig:FiveAndFour} A model consisting of two parts}
\end{figure}

\subsection{\pth-bisimilarity}

\begin{definition}[\pth-bisimilarity]\label{def:Pathbisimilarity}
Given CM $\model=(X,\closure, \peval)$ and index space
$\calJ=(I,\closure^{\calJ})$, a non empty relation 
$B \subseteq X \times X$ is a {\em \pth-bisimulation  over $X$} if, whenever $(x_1,x_2) \in B$, the following holds:
\begin{enumerate}
\item $\invpeval(x_1) = \invpeval(x_2)$;
\item for all $\pi_1 \in \bpthsF_{\calJ,\model}(x_1)$, 
there exists $\pi_2 \in \bpthsF_{\calJ,\model}(x_2)$ such that\\
$(\pi_1(\pthlen(\pi_1)), \pi_2(\pthlen(\pi_2))) \in B$;
\item for all $\pi_2 \in \bpthsF_{\calJ,\model}(x_2)$, 
there exists $\pi_1 \in \bpthsF_{\calJ,\model}(x_1)$ such that\\
$(\pi_1(\pthlen(\pi_1)), \pi_2(\pthlen(\pi_2))) \in B$;
\item for all $\pi_1 \in \bpthsT_{\calJ,\model}(x_1)$, 
there exists $\pi_2 \in \bpthsT_{\calJ,\model}(x_2)$ such that.\\
$(\pi_1(0), \pi_2(0)) \in B$;
\item for all $\pi_2 \in \bpthsT_{\calJ,\model}(x_2)$, 
there exists $\pi_1 \in \bpthsT_{\calJ,\model}(x_1)$ such that\\
$(\pi_1(0), \pi_2(0)) \in B$.
\end{enumerate}
$x_1$ and $x_2$ are {\em \pth-bisimilar}, written $x_1\,\pthbis^{\model}\, x_2$,  if and only if there is a \pth-bisimulation $B$ over $X$  such that $(x_1,x_2) \in B$. \closedefi
\end{definition}

In the sequel, we will say that two points are {\ap}-equivalent, written
 $\apeq$, if they satisfy exactly the same atomic propositions. In other words: 
$\apeq$ is the set $\ZET{(x_1,x_2)}{\invpeval(x_1) = \invpeval(x_2)}$. The following proposition trivially follows from the relevant definitions:

\begin{proposition}
\label{prp:PathImplAP}
For all \cm{s} $\model = (X,\closure, \peval)$ and $x_1, x_2 \in X$ the following holds:\\
$x_1\, \pthbis \, x_2$ implies $x_1\, \apeq \, x_2$.\qed
\end{proposition}

The converse of  \propo{} \ref{prp:PathImplAP} does not hold, as shown again in Figure~\ref{fig:TraceEqNoImplCMCbis} where we leave to the reader the easy task of
checking  that $y_{11} \not\pthbis y_{21}$ despite $y_{11} \apeq y_{21}$ since $\invpeval(y_{11}) = \invpeval(y_{21})=\SET{r}$.

In Figure~\ref{fig:PathReducedMaze} (left) an image representing a maze is shown; green pixels are the {\em exit} ones wheras the blue ones represent possible starting points; walls are represented by black pixels. In Figure~\ref{fig:PathReducedMaze} (right) the minimal model via \pth-bisimilarity is shown; it actually coincides with the one we would have obtained 
using $\apeq$ instead. In practical terms, some important features of the image of the maze are lost in its \pth-bisimilarity minimisation, such as the fact that some starting points cannot reach the exit, unless passing through walls, which should not happen! This is due to the fact that
\pth-bisimilarity abstracts from the structure of the underlying paths. In Section~\ref{sec:CoPabisimilarity} we will address this issue explicitly and refine \pth-bisimilarity into a stronger one, namely \cop-bisimilarity.
\begin{figure}
\centerline{\resizebox{4in}{!}{\includegraphics{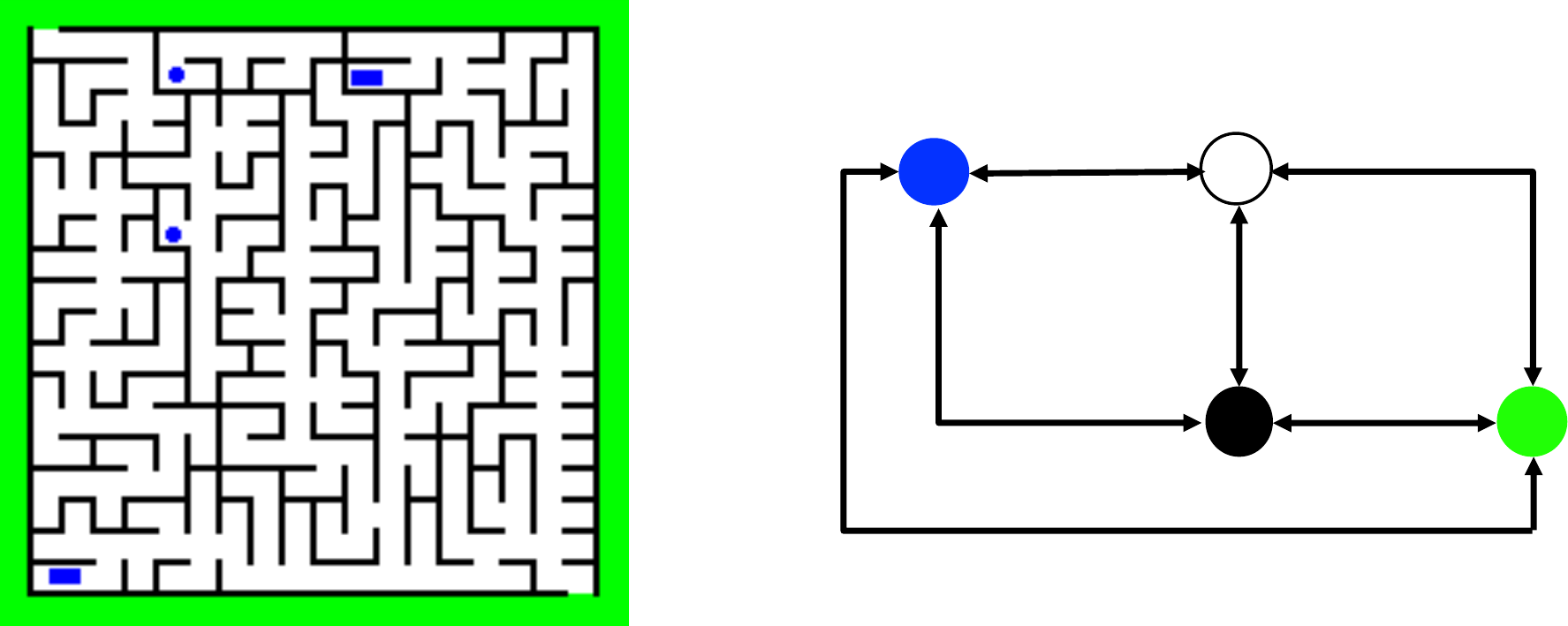}}}
\caption{\label{fig:PathReducedMaze} A maze and its reduced model modulo \pth-bisimilarity.}
\end{figure}
%

%
%
%
%
%

\begin{proposition}
\label{prp:CMCbisImplPath}
For all \qdcm{s} $\model = (X,\closure, \peval)$ and $x_1, x_2 \in X$ the following holds:
$x_1\, \cmcbis \, x_2$ implies $x_1\, \pthbis \, x_2$.
\end{proposition}

The converse of  \propo{} \ref{prp:CMCbisImplPath} does not hold, as shown in Figure~\ref{fig:PathNoImplCMCbis}
\begin{figure}[t]
\centerline{\resizebox{1.1in}{!}{\includegraphics{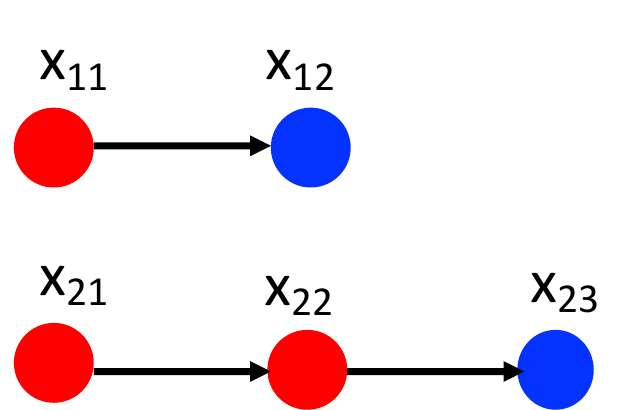}}}
\caption{\label{fig:PathNoImplCMCbis} $x_{11}  \pthbis x_{21}$ but 
$x_{11}  \not\cmcbis x_{21}$.}
\end{figure}
where 
$
\invpeval(x_{11})=\invpeval(x_{21})=\invpeval(x_{22})=\SET{r}\not=\SET{b}=\invpeval(x_{12})=\invpeval(x_{23})
$
and $x_{11}  \pthbis x_{21}$ but $x_{11}  \not\cmcbis x_{21}$ (see Remark~\ref{rem:prp:CMCbisImplPath} of Appendix~\ref{apx:sec:Pathbisimilarity}).

\begin{remark}\label{rem:CMbisNoImplPath}
It is worth pointing out that the analogous of \propo{}~\ref{prp:CMCbisImplPath} for general \cm{s} does not hold. In fact there are models with points that are \cm-bisimilar
but not \pth-bisimilar, as shown in Figure~\ref{fig:CMbisNoImplPath} where an Euclidean model $\model=(X,\closure,\peval)$ is shown such that $X=(-\infty,0) \cup (0,\infty)$,
$\closure$ is the standard closure operator for the real line $\reals$, $A,B$ and $C$ are non-empty intervals with  $B \subset A \subset (-\infty,0)$, and $C \subset   (0,+\infty)$,
$\peval(g)=A \cup C$ and  $\peval(r)=\SET{k}$, with $k\in A\setminus B$. In such a model,
$x_1 \cmbis x_2$ for all $(x_1,x_2) \in B\times C$. 
In fact $B \times C$ is a \cm-bisimulation, as shown in the sequel. Take any 
$(x_1,x_2) \in B\times C$;
clearly $\invpeval(x_1)=\invpeval(x_2)$ by construction; 
let $S_1 \subseteq (-\infty,0)$ be any set such that $x_1 \in \interior(S_1)$; then, for what concerns Condition 2 of Definition~\ref{def:CMbisimilarity}, take 
$S_2=C=\interior(C)$; for each $s_2 \in  S_2$ there is 
$s_1 \in  S_1 \cap B$ such that  
$(s_1,s_2)\in B\times C$, by definition of $B\times C$;
let finally $S_2 \subseteq (0,+\infty)$ be any set such that $x_2 \in \interior(S_2)$ then, for what concerns Condition 3 of Definition~\ref{def:CMbisimilarity}, take $S_1=B=\interior(B)$:  for each $s_1 \in  S_1$ there is $s_2 \in  S_2 \cap C$ such that  $(s_1,s_2)\in B\times C$, by definition of $B\times C$.
On the other hand, $x_1 \not\pthbis x_2$, since there cannot be any \pth-bisimulation for $x_1$ and $x_2$ as above. This is because $x_1 \in A$, since
$B \subset A$, and 
$x_1 \leadspth{} k$ with 
$r \in \invpeval(k)$ whereas  $r\not\in \invpeval(x)$ for all $x\in (0,+\infty)$.
\end{remark}

\begin{figure}
\centerline{\resizebox{4in}{!}{\includegraphics{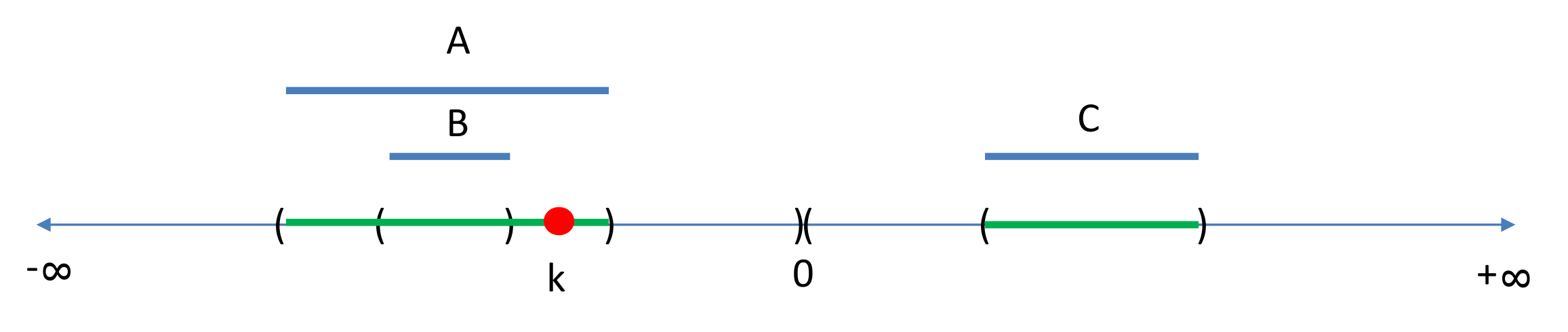}}}
\caption{\label{fig:CMbisNoImplPath} For all $x_1$ and $x_2$ such that
$(x_1,x_2)\in B \times C$ we have $x_1 \cmbis x_2$ but  $x_1 \not\pthbis x_2$.}
\end{figure}

Downstream of Remark~\ref{rem:CMbisNoImplPath} we can strengthen Definition~\ref{def:CMbisimilarity}  so that we get an adaptation for \cm{s} of the notion of 
\inl-bisimilarity proposed in~\cite{vB+17} for general neighbourhood models:

\begin{definition}[\inl-bisimilarity for \cm{s}]\label{def:INLbisimilarity}
Given CM $\model=(X,\closure, \peval)$, a non empty relation 
$B \subseteq X \times X$ is a {\em \inl-bisimulation  over $X$} if, whenever $(x_1,x_2) \in B$, the following holds:
\begin{enumerate}
\item $\invpeval(x_1) = \invpeval(x_2)$;
\item for all neighbourhoods $S_1$ of $x_1$  there is a neighbourhood $S_2$ of $x_2$ such that:
\begin{enumerate}
\item\label{back} for all $s_2 \in S_2$, there is $s_1 \in S_1$ with $(s_1,s_2)\in B$;
\item\label{forw} for all $s_1 \in S_1$, there is $s_2 \in S_2$ with $(s_1,s_2)\in B$;
\end{enumerate}
\item for all neighbourhoods $S_2$ of $x_2$ there is a neighbourhood $S_1$ of $x_1$ such that:
\begin{enumerate}
\item for all $s_1 \in S_1$, there is $s_2 \in S_2$ with $(s_1,s_2)\in B$;
\item for all $s_2 \in S_2$, there is $s_1 \in S_1$ with $(s_1,s_2)\in B$.
\end{enumerate}
\end{enumerate}
$x_1$ and $x_2$ are {\em \inl-bisimilar}, written $x_1\,\inlbis^{\model}\, x_2$,  if and only if there is a \inl-bisimulation $B$ over $X$  such that $(x_1,x_2) \in B$. \closedefi
\end{definition}

We can now prove the following

\begin{proposition}
\label{prp:MppImplInlImplpath}
For all {\em path-connected} \cm{s} $\model=(X,\closure,\peval)$ and $x_1,x_2 \in X$
the following holds: $x_1 \inlbis x_2$ implies $x_1 \pthbis x_2$.
\end{proposition}


The following proposition shows that $\pthbis$ and $\treq$ uncomparable:

\begin{proposition}
\label{prp:PathUncTraceEq}
There exist \cm{s} $\model$ and points $x_1, x_2 \in\model$ such that
$x_{1}  \pthbis x_{2}$ and $x_{1}  \not\treq x_{2}$; similarly,
there are \cm{s} $\model$ and points $x_1, x_2 \in\model$ such that
$x_{1}  \not\pthbis x_{2}$ and $x_{1}  \treq x_{2}$.
\end{proposition}

As an example of the first case, let us consider again the model of Figure~\ref{fig:PathNoImplCMCbis}: we have already seen that
$x_{11}  \pthbis x_{21}$; but $x_{11}  \not\treq x_{21}$ since
$\SET{r}\!\!\cdot\!\!\SET{b}^{\omega} \in \T (\bpthsF(x_{11})) \setminus \T (\bpthsF(x_{21}))$. As for the second case, let us consider again the model of
Figure~\ref{fig:TraceEqNoImplCMCbis}: we have already seen that $y_{11}  \treq y_{21}$
and that $y_{11}  \not\pthbis y_{21}$.

\subsection{Logical Characterisation of \pth-bisimilarity}
In this section we show that a sub-logic of \islcs{} fully characterises \pth-bisimilarity.
We first define the Infinitary Reachability Logic, \irl{} for short and show that \irl{} is a sub-logic of \islcs{} obtained by forcing the second argument of $\ltothru$ and $\lfromthru$ to $\ltrue$. Then we provide the characterisation result.

\begin{definition}[Infinitary Reachability Logic - \irl{}]\label{def:Irl}
For index set $I$ and $p \in  \ap$ the abstract language of \irl{} is defined as follows:
$$
\form ::= p \; \sep \; \lneg \form \; \sep \; \liand_{i \in I} \form_i \; \sep \; \lto \form \; \sep \; \lfrom \form.
$$

The satisfaction relation for all \cm{s}  $\model$, $x \in \model$, and \irl{} formulas $\form$ is defined recursively on the structure of $\form$ as follows:
$$
\begin{array}{r c l c l c l l}
\model,x & \models_{\irl}  & p & \Leftrightarrow & x  \in \peval(p);\\
\model,x & \models_{\irl}  & \lneg \,\form & \Leftrightarrow & \model,x  \models_{\irl} \form \mbox{ does not hold};\\
\model,x & \models_{\irl}  & \liand_{i\in I} \form_i  & \Leftrightarrow &
\model,x  \models_{\irl} \form_i \mbox{ for all } i \in I;\\
\model,x  & \models_{\irl} & \lto \form& \Leftrightarrow &
\mbox{there exist  path } \pi \mbox{ and index } \ell \mbox{ such that }\\
&&&&\mbox{\hspace{0.1in}} \pi(0) = x \mbox{ and }\pi(\ell) \models_{\irl} \form;\\
\model,x  & \models_{\irl} & \lfrom \form& \Leftrightarrow &
\mbox{there exist  path } \pi \mbox{ and index } \ell \mbox{ such that }\\
&&&&\mbox{\hspace{0.1in}} \pi(\ell) = x \mbox{ and } \pi(0) \models_{\irl} \form.\\
\end{array}
$$
\closedefi
\end{definition}

The following proposition trivially follows from the relevant definitions:

\begin{proposition}\label{prp:ToFromToThrouETC}
For all \cm{s} $\model = (X,\closure, \peval)$ and $x_1, x_2 \in X$ the following holds:\\
$\lto \form \equiv \ltothru \form [\ltrue]$ and
$\lfrom \form \equiv \lfromthru \form [\ltrue]$.\qed
\end{proposition}

\begin{definition}[\irl-Equivalence]\label{def:IrlEq}
Given \cm{} $\model = (X,\closure,\peval)$, the equivalence relation $\irleq^{\model} \, \subseteq \, X \times X$ is defined as: $x_1 \irleq^{\model} x_2$ if and only if for all \irl{} formulas $\form$ the following holds: $\model, x_1 \models_{\irl} \form$ if and only if $\model, x_2 \models_{\irl} \form$.
\closedefi
\end{definition}

\begin{theorem}\label{thm:PathbisIsIrleq}
For all \cm{s} $\model=(X,\closure,\peval)$, any
\pth-Bisimulation   $B$ over $X$ is included in the equivalence $\irleq^{\model}$.
\closethm
\end{theorem}

The  converse of Theorem~\ref{thm:PathbisIsIrleq} is given below.

\begin{theorem}\label{thm:IIrleqIsPathbis}
For all \cm{s}  $\model=(X,\closure,\peval)$, $\irleq^{\model}$ is a \pth-bisimulation.
\closethm
\end{theorem}

\begin{corollary}\label{cor:PathbisEqIrleq}
For all \cm{s}  $\model=(X,\closure,\peval)$ we have that $\irleq^{\model}$ coincides with $\pthbis^{\model}$.
\qed 
\end{corollary}

\section{\cop-bisimilarity}\label{sec:CoPabisimilarity}

\pth-bisimilarity is in some sense too weak, too abstract; nothing whatsoever is required of the relevant paths, except their starting points being fixed and related by the bisimulation, and  their end-points be in the bisimulation as well. A  deeper insight into the structure of such paths would be desirable as well as some, relatively high level, requirements over them. To that purpose we resort to a notion of ``compatibility'' between relevant paths that essentially requires each of them to be composed of a non-empty sequence of non-empty, adjacent ``zones''.
More precisely, both paths under consideration in a transfer condition should share the same structure, as follows (see Figure~\ref{fig:Zones}):
\begin{itemize}
\item both paths are composed by a sequence of (non-empty) ``zones'';
\item the number of zones should be the same in both paths, {\em but}
\item the length of ``corresponding'' zones can be different, {\em as well as}
 the length of the two paths;
\item {\em each} point in one zone of a path should be related by the bisimulation to {\em every} point in the corresponding zone of  the other path.
\end{itemize}
This notion of compatibility gives rise to {\em Compatible Path bisimulation}, \cop-bisimulation, defined below.
We note that the notion of \cop-bisimulation  turns out to be reminiscent of that of {\em Equivalence with respect to Stuttering} for transition systems proposed in~\cite{BCG88}, although in a totally different context and with a quite different definition: the latter is defined via a convergent sequence of relations and makes use of a different notion of path than the one of \cs{} used in this paper. Finally, \cite{BCG88} is focussed on CTL/CTL$^*$, which 
implies a flow of time with single past (i.e. trees), which is not the case for structures representing space.

\begin{figure}
\centerline{\resizebox{3in}{!}{\includegraphics{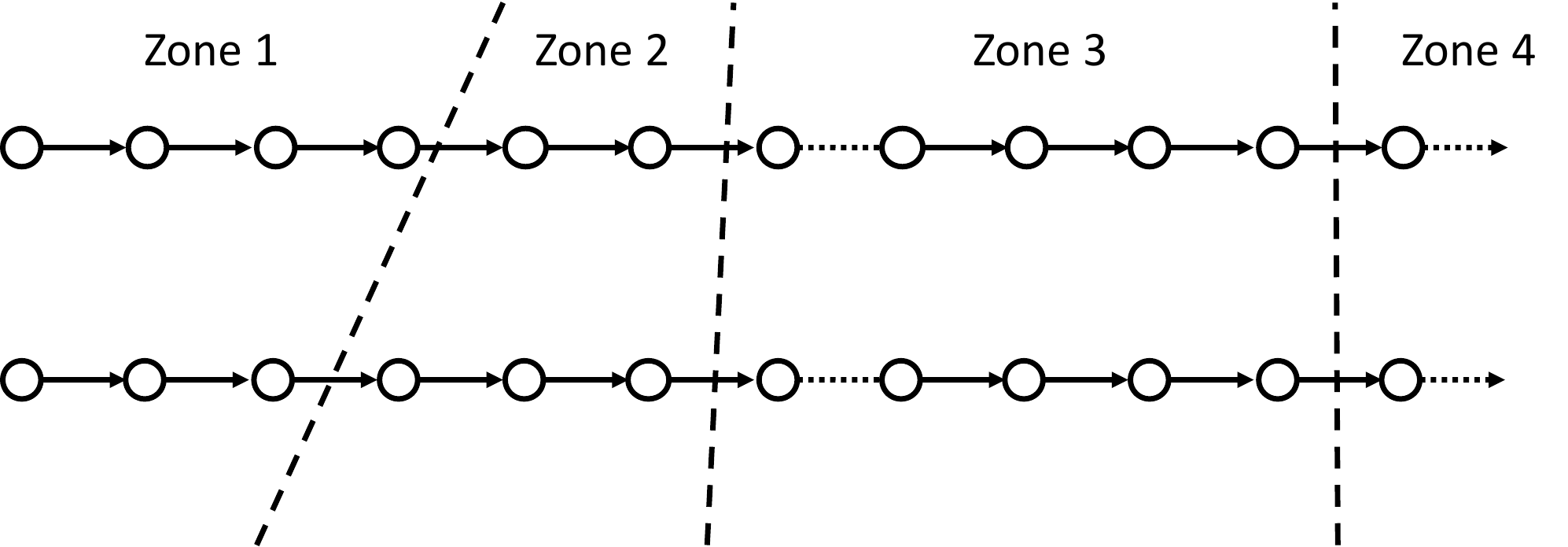}}}
\caption{\label{fig:Zones} Zones in relevant paths.}
\end{figure}

\subsection{\cop-bisimilarity}

\begin{definition}[\cop-bisimilarity]\label{def:CoPabisimilarity}
Given CM $\model=(X,\closure, \peval)$ and index space
$\calJ=(I,\closure^{\calJ})$, a non empty relation 
$B \subseteq X \times X$ is a {\em \cop-bisimulation  over $X$} if, whenever $(x_1,x_2) \in B$, the following holds:
\begin{enumerate}
\item $\invpeval(x_1) = \invpeval(x_2)$;
\item\label{zones} for all $\pi_1 \in \bpthsF_{\calJ,\model}(x_1)$ such that\\ 
$(\pi_1(i_1),x_2) \in B$ for all $i_1 \in \ZET{\iota}{0 \leq \iota < \pthlen(\pi_1)}$,\\
there is
$\pi_2 \in \bpthsF_{\calJ,\model}(x_2)$ such that the following holds:\\
$(x_1,\pi_2(i_2)) \in B$ for all $i_2\in \ZET{\iota}{0 \leq \iota < \pthlen(\pi_2)}$, and\\
$(\pi_1(\pthlen(\pi_1)), \pi_2(\pthlen(\pi_2))) \in B$;
\item for all $\pi_2 \in \bpthsF_{\calJ,\model}(x_2)$ such that\\
$(x_1,\pi_2(i_2)) \in B$ for all $i_2 \in \ZET{\iota}{0 \leq \iota < \pthlen(\pi_2)}$,\\
there is $\pi_1 \in \bpthsF_{\calJ,\model}(x_1)$ such that the following holds:\\
$(\pi_1(i_1),x_2) \in B$ for all $i_1\in \ZET{\iota}{0 \leq \iota < \pthlen(\pi_1)}$, and\\
$(\pi_1(\pthlen(\pi_1)),\pi_2(\pthlen(\pi_2))) \in B$;
\item for all $\pi_1 \in \bpthsT_{\calJ,\model}(x_1)$ such that \\
$(\pi_1(i_1),x_2) \in B$ for all $i_1 \in \ZET{\iota}{0 < \iota \leq \pthlen(\pi_1)}$,\\
there is $\pi_2 \in \bpthsT_{\calJ,\model}(x_2)$ such that the following holds:\\
$(x_1,\pi_2(i_2)) \in B$ for all $i_2\in \ZET{\iota}{0 < \iota \leq \pthlen(\pi_2)}$, and\\
$(\pi_1(0), \pi_2(0)) \in B$;
\item for all $\pi_2 \in \bpthsT_{\calJ,\model}(x_2)$ such that \\
$(x_1,\pi_2(i_2)) \in B$ for all $i_2 \in \ZET{\iota}{0 < \iota \leq \pthlen(\pi_2)}$,\\
there is $\pi_1 \in \bpthsT_{\calJ,\model}(x_1)$ such that the following holds:\\
$(\pi_1(i_1),x_2) \in B$ for all $i_1\in \ZET{\iota}{0 < \iota \leq \pthlen(\pi_1)}$, and\\
$(\pi_1(0),\pi_2(0)) \in B$;
\end{enumerate}
$x_1$ and $x_2$ are {\em \cop-bisimilar}, written $x_1\,\copbis^{\model}\, x_2$,  if there is a \cop-bisimulation $B$ over $X$  such that $(x_1,x_2) \in B$. \closedefi
\end{definition}

Figure~\ref{fig:CoPaReducedMaze} shows the minimal model modulo \cop-bisimilarity of the maze image shown in Figure~\ref{fig:PathReducedMaze}. It is easy to see that this reduced model retains more information than that of  Figure~\ref{fig:PathReducedMaze} (right). In particular, in this model three different representatives of white points are present: 
\begin{itemize}
\item one that is directly connected both with a representative of a blue starting point and with a representative of a green exit point; this represents the situation in which from a blue starting point the exit can be reached walking through the maze (i.e. white points);
\item one that is directly connected with  a representative of a green point, but it is not directly connected with a representative of a blue point; this represents parts of the maze from which an exit could be reached, but that are separated (by walls) from areas where there are starting points (see below), and
\item  one that is directly connected to a representative of a blue starting point but that is not directly connected to a green exit point---that can be reached only by passing through the black point; this represents the fact that the relevant blue starting point cannot reach the exit because it will always be blocked by a wall.
\end{itemize}

\begin{figure}
\centerline{\resizebox{2in}{!}{\includegraphics{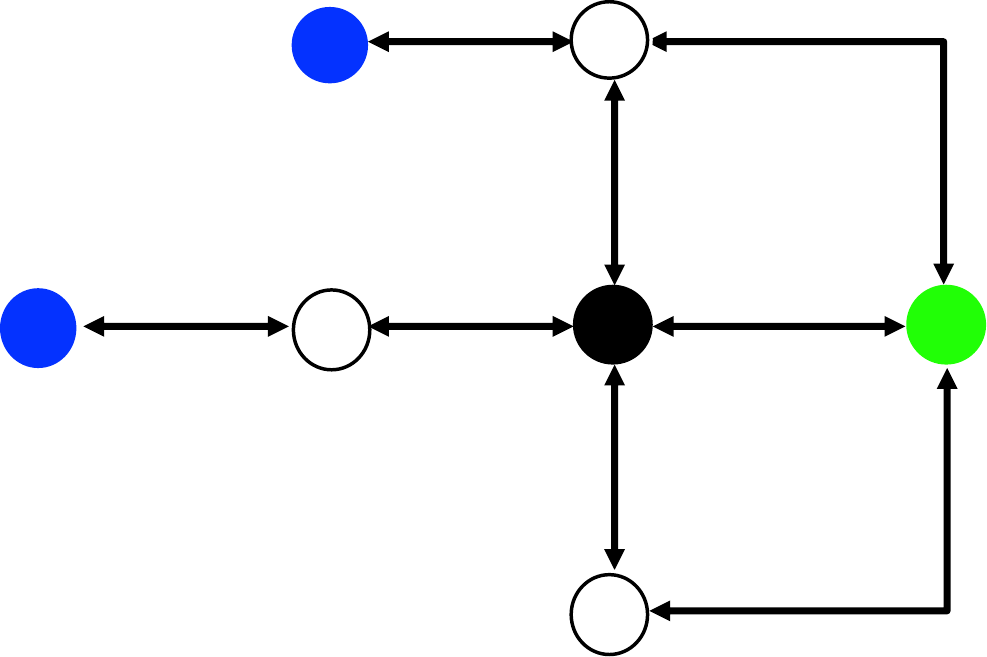}}}
\caption{\label{fig:CoPaReducedMaze} Reduced model of the maze of Fig.~\ref{fig:PathNoImplCMCbis} (left), modulo \cop-bisimilarity.}
\end{figure}
%

%
%


The following proposition can be easily proved from the relevant definitions:

\begin{proposition}
\label{prp:CoPaImplPath}
For all \cm{s} $\model = (X,\closure, \peval)$ and $x_1, x_2 \in X$ the following holds:
$x_1\, \copbis \, x_2$ implies $x_1\, \pthbis \, x_2$.\qed
\end{proposition}

The converse of  \propo{} \ref{prp:CoPaImplPath} does not hold, as shown  in Figure~\ref{fig:PathNoImplCoPa}. Relation 
$B=\SET{(t_{11},t_{21}), (t_{12},t_{22}), (t_{13},t_{23}), (t_{14},t_{24}), (t_{15},t_{25})}$ is a \pth-bisimulation, so $t_{11} \pthbis t_{21}$. On the other hand, Condition 2 of 
Definition~\ref{def:CoPabisimilarity} cannot be fulfilled for any $\pi_1 \in \bpthsF(t_{11})$ such that  $\pi_1(j)=t_{13}$ for some $j>0$ since for every
$\pi_2\in \bpthsF(t_{21})$ there is $k$ such that $g \in \invpeval(\pi_2(k))$, wheras 
$g \not\in \invpeval(\pi_1(h))$ for all $h$.
\begin{figure}
\centerline{\resizebox{2in}{!}{\includegraphics{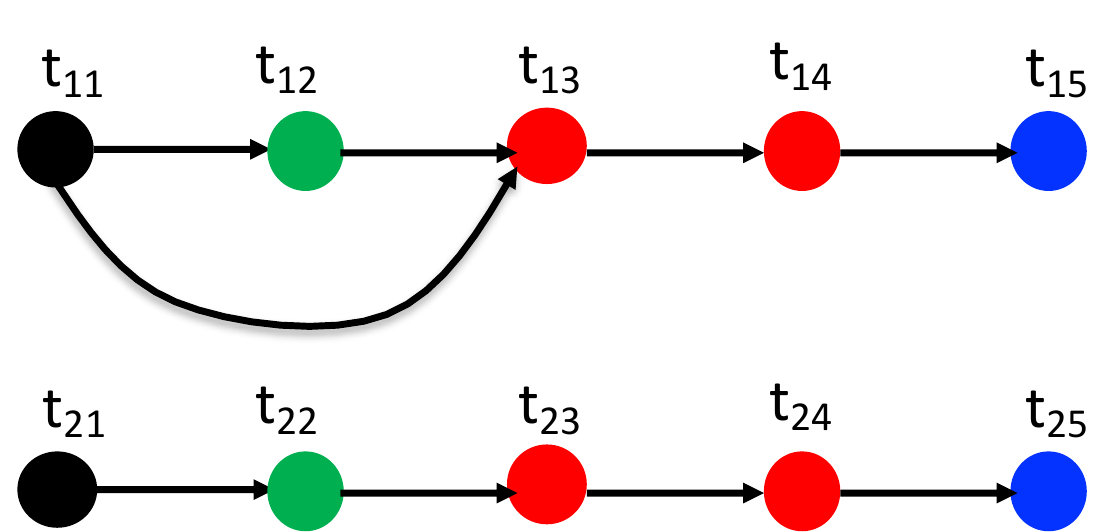}}}
\caption{\label{fig:PathNoImplCoPa} $t_{11} \pthbis t_{21}$ but $t_{11} \not\copbis t_{21}$.}
\end{figure}

\begin{proposition}
\label{prp:CMCbisImplCoPa}
For all \qdcm{s} $\model = (X,\closureF, \peval)$ and $x_1, x_2 \in X$ the following holds:
$x_1\, \cmcbis \, x_2$ implies $x_1\, \copbis \, x_2$.
\end{proposition}

The converse of  \propo{} \ref{prp:CMCbisImplCoPa} does not hold; again with reference to Figure~\ref{fig:PathNoImplCMCbis},  it is easy to see that  $B=\SET{(x_{11},x_{21}), (x_{11},x_{22}), (x_{12},x_{23})}$ is a \cop-bisimulation, and so 
$x_{11} \copbis x_{21}$. On the other hand, as we have already seen, $x_{11} \not\cmcbis x_{21}$.

The following proposition shows that $\copbis$ and $\treq$ are uncomparable.

\begin{proposition}
\label{prp:CoPaUncTraceEq}
There exist \cm{s} $\model$ and points $x_1, x_2 \in\model$ such that
$x_{1}  \copbis x_{2}$ and $x_{1}  \not\treq x_{2}$; similarly,
there are \cm{s} $\model$ and points $x_1, x_2 \in\model$ such that
$x_{1}  \not\copbis x_{2}$ and $x_{1}  \treq x_{2}$.
\end{proposition}

As an example of the first case, let us consider again the model of Figure~\ref{fig:PathNoImplCMCbis}: we have already seen that
$x_{11}  \copbis x_{21}$  and that $x_{11}  \not\treq x_{21}$. 
As for the second case, let us consider again the model of
Figure~\ref{fig:TraceEqNoImplCMCbis}: we have already seen that
$x_{11}  \not\pthbis x_{21}$, and thus, by \propo{}~\ref{prp:CoPaImplPath} we
get $x_{11}  \not\copbis x_{21}$; but we have  already seen that
$x_{11}  \treq x_{21}$.

\subsection{\cop-bisimilarity minimisation}\label{sec:CopMin}
In this section we show how \cop-bisimilarity minimisation can be achieved using results from~\cite{Gr+17} on 
minimisation of {\em Divergence-blind Stuttering Equivalence}.
We first recall the definition of Divergence-blind Stuttering Equivalence (Def. 2.2 of \cite{Gr+17}):
\begin{definition}{Divergence-blind Stuttering Equivalence (\dbs-Eq)}\label{def:dbs}. \\
Let $K = (S,\ap{}, R, L)$ be a Kripke structure. A symmetric relation $E \subseteq S \times S$ is a divergence-blind stuttering equivalence if and only if for all $s, t \in S$ such that  $s \, E \, t$:
\begin{enumerate}
\item $L(s) = L(t)$, and
\item for all $s'  \in S$, if $s \,R \,  s'$ then there are 
$t_0, \ldots , t_k \in S$ for some $k \in \nats$ such that 
$t = t_0, s \,E \, t_i , t_i \,R \, t_{i+1}$ for all $i < k$, and  $s' \, E \, t_k$.
\end{enumerate}
We say that two states $s, t\in S$ are divergence-blind stuttering equivalent, notation
$s \dbseq t$, if and only if there is a divergence-blind stuttering equivalence relation $E$
such that $s \, E \, t$. \closedefi
\end{definition}
First of all we recall that every Kripke structure $K = (S,\ap, R, L)$ gives rise to a \qdcm{,} namely the model $\model(K) = (S,\closure_{R}, \peval_L)$ where 
$\peval_L(p)=\ZET{s \in S}{p \in L(s)}$. Similarly, every \qdcm{}  
$\model = (S,\closure_{R}, \peval)$ characterises a Kripke structure 
$K(\model) = (S,\ap, R, L_{\peval})$ where $L_{\peval}(s)=\invpeval(s)$.
We also recall that a {\em path} in $K$ is a
sequence $s_0, \ldots , s_k \in S$ such that $s_i \,R \, s_{i+1}$ for all $i < k$. Note that this definition of path is different from that of path in a \qdcm{.} For instance, consider 
Kripke structure $(\SET{s,t},\ap, \SET{(s,t)}, L)$, for some $s\not=t$ and $L$,
and related \qdcm{}
$(\SET{s,t}, \closure_{\SET{(s,t)}}, \peval_L)$. In the Kripke structure
there is no path corresponding to the following path in the \qdcm{:}
$\pi(0)=\pi(1)= s$, and $\pi(n+2)=t$ for all $n \in \nats$, and this is because $(s,s) \not\in \SET{(s,t)}$. In other words, paths in Kripke structures are strictly bound to the 
accessibility relation of the structure, while those in \qdcm{} are more flexible in this respect, due to their possibility of having more adjacent indexes being mapped to the same point (i.e. ``stuttering''). Of course, for each Kripke structure $K = (S,\ap, R, L)$ there is a 
Kripke structure $K^r$ having exactly the same paths as those of  $\model(K)$, namely $K^r = (S,\ap, R^r, L)$, where, we recall, $R^r$ is the reflexive closure of $R$. Note, by the way, $\model(K^r)= \model(K)$, i.e. $K$ and $K^r$ share the same \qdcm{.} This is due to the
fact that $\closure_R = \closure_{R^r}$ and is a consequence of the very definition of $\closure$.

We now provide a ``back-and-forth'' version of \dbs-Eq:

\begin{definition}[Divergence-blind Stuttering Equivalence with Converse (\dbsc-Eq)]\label{def:dbsc}
Let $K = (S,\ap{}, R, L)$ be a Kripke structure. A symmetric relation $E \subseteq S \times S$ is a divergence-blind stuttering equivalence with converse if and only if 
for all $s, t \in S$ such that  $s \, E \, t$:
\begin{enumerate}
\item $L(s) = L(t)$, and
\item for all $s'  \in S$, if $s \,R \,  s'$, then there are 
$t_0, \ldots , t_k \in S$ for some $k \in \nats$ such that 
$t_0=t, s \,E \, t_i , t_i \,R \, t_{i+1}$ for all $i\in\ZET{\iota}{0\leq \iota < k}$, 
and  $s' \, E \, t_k$;
\item for all $s'  \in S$, if $s' \,R \,  s$, then there are 
$t_0, \ldots , t_k \in S$ for some $k \in \nats$ such that 
$t_k=t, s \,E \, t_i , t_{i-1} \,R \, t_{i}$ for all $i\in\ZET{\iota}{0< \iota \leq k}$, 
and  $s' \, E \, t_0$.
\end{enumerate}
We say that two states $s, t\in S$ are divergence-blind stuttering with converse equivalent, notation
$s \dbsceq t$, if and only if there is a divergence-blind stuttering equivalence with converse relation $E$
such that $s \, E \, t$. \closedefi
\end{definition}

%
%
\begin{proposition}\label{prp:DbsEQCopa}
For every \qdcm{} $\model=((X,\closure_R),\peval)$, $x_1, x_2 \in X$
$x_1 \, \copbis \, x_2$ with respect to $\model$ if and only if 
$x_1 \, \dbsceq \, x_2$ with respect to $K(\model)^r$.
\end{proposition}
\propo{} \ref{prp:DbsEQCopa}  gives an effective way for computing the minimisation of $\model$ w.r.t.
$\copbis$, by using the algorithm(s) proposed in~\cite{Gr+17}.

\subsection{Logical Characterisation of \cop-bisimilarity}
In this section we show that a sub-logic of \islcs{} fully characterises \cop-bisimilarity.
We first define the Infinitary Compatible Reachability Logic, \icrl{} for short and show that \icrl{} is a sub-logic of \islcs{} obtained by forcing $\ltothru$ and $\lfromthru$ 
to be used only in conjunction of their second argument. Then we provide the characterisation result.

\begin{definition}[Infinitary Compatible Reachability Logic - \icrl]\label{def:Icrl}
For index set $I$ and $p \in  \ap$ the abstract language of \icrl{} is defined as follows:
$$
\form ::= p \; \sep \; \lneg \form \; \sep \; \liand_{i \in I} \form_i \; \sep \; \lstothru \form_1[\form_2] \; \sep \; \lsfromthru \form_1[\form_2].
$$

The satisfaction relation for all \cm{s}  $\model$, $x \in \model$, and \icrl{} formulas $\form$ is defined recursively on the structure of $\form$ as follows:
$$
\begin{array}{r c l c l c l l}
\model,x & \models_{\icrl}  & p & \Leftrightarrow & x  \in \peval(p);\\
\model,x & \models_{\icrl}  & \lneg \,\form & \Leftrightarrow & \model,x  \models_{\icrl} \form \mbox{ does not hold};\\
\model,x & \models_{\icrl}  & \liand_{i\in I} \form_i  & \Leftrightarrow &
\model,x  \models_{\irl} \form_i \mbox{ for all } i \in I;\\
\model,x  & \models_{\icrl} & \lstothru \form_1 [\form_2] & \Leftrightarrow &
\mbox{there exist  path } \pi \mbox{ and index } \ell \mbox{ such that }\\
&&&&\mbox{\hspace{0.1in}} \pi(0) = x \mbox{ and }\\
&&&&\mbox{\hspace{0.1in}} \pi(\ell) \models_{\icrl} \form_1 \mbox{ and }\\
&&&&\mbox{\hspace{0.1in}} \mbox{for all } j \mbox{ such that } 0 \leq j < \ell \mbox{ the following holds:}\\
&&&&\mbox{\hspace{0.2in}} \pi(j) \models_{\icrl} \form_2;\\
\model,x  & \models_{\icrl} & \lsfromthru \form_1 [\form_2] & \Leftrightarrow &
\mbox{there exist  path } \pi \mbox{ and index } \ell \mbox{ such that }\\
&&&&\mbox{\hspace{0.1in}} \pi(\ell) = x \mbox{ and }\\
&&&&\mbox{\hspace{0.1in}} \pi(0) \models_{\icrl} \form_1 \mbox{ and }\\
&&&&\mbox{\hspace{0.1in}} \mbox{for all } j \mbox{ such that } 0  < j \leq \ell \mbox{ the following holds:}\\
&&&&\mbox{\hspace{0.2in}} \pi(j) \models_{\icrl} \form_2.
\end{array}
$$
\closedefi
\end{definition}

The following proposition trivially follows from the relevant definitions:

\begin{proposition}\label{prp:SToThrouToThrouETC}
For all \cm{s} $\model = (X,\closure, \peval)$ and $x_1, x_2 \in X$ the following holds:
$\lstothru \form_1 [\form_2] \equiv \form_2 \land \ltothru \form_1 [\form_2]$ and
$\lsfromthru \form_1 [\form_2] \equiv \form_2 \land \lfromthru \form_1 [\form_2]$.\qed
\end{proposition}

\begin{definition}[\icrl-Equivalence]\label{def:IcrlEq}
Given \cm{} $\model = (X,\closure,\peval)$, the equivalence relation $\icrleq^{\model} \, \subseteq \, X \times X$ is defined as: $x_1 \icrleq^{\model} x_2$ if and only if for all \icrl{} formulas $\form$, it holds: $\model, x_1 \models_{\icrl} \form$ if and only if $\model, x_2 \models_{\icrl} \form$.
\closedefi
\end{definition}

\begin{theorem}\label{thm:CoPabisIsIcrleq}
For all \qdcm{s} $\model=(X,\closureF,\peval)$, any
\cop-bisimulation   $B$ over $X$ is included in the equivalence $\icrleq^{\model}$.
\closethm
\end{theorem}

The  converse of Theorem~\ref{thm:CoPabisIsIcrleq} is given below.

\begin{theorem}\label{thm:IcrleqIsCoPabis}
For all \qdcm{s}  $\model=(X,\closureF,\peval)$, $\icrleq^{\model}$ is a \cop-bisimulation.
\closethm
\end{theorem}

\begin{corollary}\label{cor:CoPabisEqIcrleq}
For all \qdcm{s}  $\model=(X,\closureF,\peval)$ we have that $\icrleq^{\model}$ coincides with $\copbis^{\model}$.
\qed 
\end{corollary}

\section{Conclusions}\label{sec:Conclusions}

In this paper we have studied three main bisimilarities for closure spaces, namely
\cm-bisimilarity, and its specialisation for \qdcm{s} \cm-bisimilarity with converse, \pth-bisimilarity, and \cop-bisimilarity.

\cm-bisimilarity is a generalisation for \cm{s} of classical Topo-bisimilarity for topological spaces.
\cm-bisimilarity with converse takes into consideration the fact that, in \qdcm{s}, there is a notion of ``direction'' given by the binary relation underlying the closure operator. This can be exploited in order to get an equivalence---namely \cm-bisimilarity with converse---that, for \qdcm{s}, refines \cm-bisimilarity. We have shown that 
\cm-bisimilarity with converse coincides with \cl-bisimilarity defined~\cite{Ci+20}. 
Both \cm-bisimilarity and \cm-bisimilarity with converse turn out to be too strong for expressing interesting properties of spaces. To that purpose we introduce \pth-bisimilarity that characterises
unconditional reachability in the space, and a stronger equivalence, \cop-bisimilarity, that expresses a notion of path ``compatibility'' resembling the concept of {\em stuttering} equivalence for transition systems~\cite{BCG88}.

For each notion of bisimilarity we also provide a modal logic that characterises it. We finally address the issue of space minimisation via bisimulation and provide a recipe for \cop-bisimilarity minimisation; minimisation via \cm-bisimilarity with converse has already been dealt with in~\cite{Ci+20} whereas minimisation via \pth-bisimilarity is a special case of that via \cop-bisimilarity (also note that $\lto \, \form \equiv \lstothru \form [\ltrue]$ and, similarly, 
$\lfrom \, \form \equiv \lsfromthru \form [\ltrue]$).

Many results we have shown in this paper concern \qdcm{s;} we think the investigation of their extension to continuous or  general closure spaces is an interesting line of future research. In~\cite{Ci+20} we investigated a coalgebraic view of  \qdcm{s} that was useful for the definition
of the minimisation algorithm for \cl-bisimilarity. It would be interesting to study a similar approach for \pth-bisimilarity and \cop-bisimilarity.

\bibliographystyle{splncs03}
\bibliography{abbr,clmv,xref}

\appendix

\section{Proofs of Results of Section~\ref{sec:Preliminaries}}\label{apx:sec:Preliminaries}
\subsection{Proof of  \propo{} \ref{prp:FT}}\label{prf:prp:FT}
We prove only Point~\ref{path} of the proposition, the proof of the other points being trivial.
We show that $\pi$ is a path over $X$ if and only if, for all  
$i \in (\dom \, \pi) \setminus \SET{0}$, we have $ \pi(i) \in \closureF(\pi(i-1))$.
Suppose $\pi$ is a path over $X$; the following derivation, valid for all $i\in \nats$, proves the assert:\\
$ $\\
\noindent
$
\deriv
\pi(i)
\hint{\in}{Set Theory}
\SET{\pi(i-1),\pi(i)}
\hint{=}{Definition of $\pi(N)$ for $N \subseteq \nats$}
\pi(\SET{i-1,i})
\hint{=}{Definition of $\closure_{\succ}$}
\pi(\closure_{\succ}(\SET{i-1}))
\hint{\subseteq}{Continuity of $\pi$}
\closureF(\pi(i-1))
$

For proving the converse we have to show that for all sets 
$N \subseteq (\dom \, \pi)$ we have 
$\pi(\closure_{\succ}(N)) \subseteq \closureF(\pi(N))$. 
By definition of $\closure_{\succ}$
we have that $\closure_{\succ}(N) = N \cup \ZET{i}{i-1 \in N}$ and so
$\pi(\closure_{\succ}(N)) = \pi(N) \cup \pi(\ZET{i}{i-1 \in N})$. 
By the second axiom of closure, we have $\pi(N) \subseteq \closureF(\pi(N))$.
We show that $\pi(\ZET{i}{i-1 \in N}) \subseteq \closureF(\pi(N))$ as well.
Take any $i$ such that $i-1 \in N$; we have $\SET{\pi(i-1)} \subseteq \pi(N)$ since $i-1 \in N$, and,
by monotonicity of $\closureF$ it follows that 
$\closureF(\SET{\pi(i-1)}) \subseteq \closureF(\pi(N))$ and since 
$\pi(i) \in \closureF(\pi(i-1))$ by hypothesis, we also get 
$\pi(i) \in\closureF(\pi(N))$. Since this holds for all elements of
the set $\ZET{i}{i-1 \in N}$ we also have 
$\pi(\ZET{i}{i-1 \in N}) \subseteq \closureF(\pi(N))$.

The proof for $ \pi(i-1) \in \closureT(\pi(i))$ is similar.

\section{Proofs of Results of Section~\ref{sec:CMbisimilarity}}\label{apx:sec:CMbisimilarity}

\subsection{Proof of \propo{} \ref{prp:HomeoImplCMbis}}\label{prf:prp:HomeoImplCMbis}
We show that $\homeo$ is a \cm-bisimulation. Suppose, \wlg, that $x_2 = h(x_1)$ for some homeomorphism $h: X \to X$.
Condition 1 of Definition~\ref{def:CMbisimilarity} is trivially satisfied due to Condition 1
of Definition~\ref{def:Homeo}. For what concerns Condition 2 of Definition~\ref{def:CMbisimilarity}, let $S_1$ a neighbourhood of $x_1$. Define
$S_2$ as $S_2=h(S_1)$. We have 
$
x_2 = h(x_1) \in h(\interior(S_1))=\interior(h(S_1))=\interior(S_2)
$, where in the one but last step we exploited Condition 3 of Definition~\ref{def:Homeo}.
Now we can easily see that  Condition 2 of Definition~\ref{def:CMbisimilarity} is satisfied since, by definition of $S_2$, for all $s_2 \in S_2$ there exists $s_1=h^{-1}(s_2) \in S_1$ such that $s_2=h(s_1)$, i.e. $s_1 \homeo s_2$. The proof for Condition 3 of  
Definition~\ref{def:CMbisimilarity} is similar.\\

\begin{remark}\label{rem:prp:HomeoImplCMbis}
The converse of \propo{} \ref{prp:HomeoImplCMbis} does not hold, as shown in Figure~\ref{fig:CMbisNoImplHomeo} where 
$
\invpeval(x_{11})=\invpeval(x_{21})=\SET{r}\not=\SET{b}=\invpeval(x_{12})=\invpeval(x_{22})=\invpeval(x_{23})
$
and $x_{11}  \cmbis x_{21}$ but $x_{11}  \not\homeo x_{21}$. In fact, any non-trivial homeomorphism $h$ should map $x_{11}$ to $x_{21}$ (or viceversa), and any
of $x_{12}$, $x_{22}$ and $x_{23}$ to any of $x_{12}$, $x_{22}$ and $x_{23}$,
otherwise Condition 1 of Definition~\ref{def:Homeo} would be violated. In addition, 
in order not to violate injectivity, $h$ hould be a permutation over 
$\SET{x_{12}, x_{22}, x_{23}}$. Let us suppose, \wlg, 
$h(x_{11})=x_{21}$ and $h(x_{12})=x_{22}$. Then we would get
$
h(\closure(\SET{x_{11}}))= h(\SET{x_{11}, x_{12}}) = \SET{x_{21},x_{22}}\not=
\SET{x_{21},x_{22},x_{23}}=\closure(\SET{x_{21}})=\closure(h(\SET{x_{11}}))
$, violating (the equivalent of) Condition 3 of Definition~\ref{def:Homeo}.
\end{remark}

\subsection{Proof of \theo{} \ref{thm:CMbisIsImleq}}\label{prf:thm:CMbisIsImleq}

We proceed by induction on the structure of $\form$ and 
consider only the case  $\form = \lnear \form'$, the others being trivial.
Suppose $B$ is a \cm-Bisimulation,  $(x_1,x_2) \in B$ and, \wlg, 
$\model,x_1 \not\models \lnear \form'$ and $\model,x_2 \models \lnear \form'$, that is
$x_2 \in \closure(\sem{\form'})$ and
$
x_1 \in \overline{\closure(\sem{\form'})}=
\overline{\overline{\interior (\overline{\sem{\form'}})}}=
\interior(\overline{\sem{\form'}})
$.

Let $S_1=\overline{\sem{\form'}}$ and, by $x_1 \in \interior(\overline{\sem{\form'}})$,
let $S_2$ be chosen according to Definition~\ref{def:CMbisimilarity}, with $x_2 \in \interior(S_2)$. By Lemma~\ref{lem:IntIfClImInt} below, we have $\sem{\form'}\cap S_2 \not=\emptyset$, since $x_2 \in \closure(\sem{\form'})  \cap \interior(S_2)$. Let thus $s_2$ belong to $\sem{\form'} \cap S_2$ and since $B$ is a \cm-Bisimulation, there exists $s_1 \in S_1$ such that  $(s_1,s_2) \in B$ (Condition 2 of Definition~\ref{def:CMbisimilarity}), with $s_2 \in \sem{\form'}$---by definition of $s_2$. By the induction hypothesis, since $\model, s_2 \models \form'$ and $(s_1,s_2) \in B$ we get
$\model, s_1 \models \form'$ which contradicts $s_1\in \overline{\sem{\form'}} = S_1$.

\begin{lemma}\label{lem:IntIfClImInt}
For all \cm{s} $\model=(X,\closure,\peval)$, for all  $Y,Z \subseteq X$ the following holds: if 
$\closure(Y) \cap \interior(Z) \not=\emptyset$ then $Y \cap Z \not=\emptyset$.
\end{lemma}

\begin{proof}
We prove that $Y \cap Z =\emptyset$ implies $\interior(Z) \cap  \closure(Y) = \emptyset$.
Suppose $Y \cap Z =\emptyset$. 
Then $Y \subseteq  \overline{Z}$, 
and so  $\closure(Y) \subseteq  \closure(\overline{Z})$, 
that is $\overline{\closure(\overline{Z})} \subseteq \overline{\closure(Y)}$.
So $\interior(Z) \subseteq \overline{\closure(Y)}$, 
that is $\interior(Z) \cap  \closure(Y) = \emptyset$.
This proves the assert.
\end{proof}

\subsection{Proof of \theo{} \ref{thm:ImleqIsCMbis}}\label{prf:thm:ImleqIsCMbis}

The following proof has been inspired by the proof of an analogous theorem 
in~\cite{vB+17}. We first need a preliminary definition:

\begin{definition}\label{def:deltachi}
Given \cm{}  $\model=(X,\closure,\peval)$, for all $x_1,x_2 \in X$, 
let formula $\delta_{x_1,x_2}$ be defined as follows: if $x_1 \imleq x_2$, then set 
$\delta_{x_1,x_2}$ to $\ltrue$; otherwise, choose a formula, say $\Phi_{x_1,x_2}$, such that $\model,x_1 \models \Phi_{x_1,x_2}$ and
$\model, x_2 \models \lneg\Phi_{x_1,x_2}$ and set $\delta_{x_1,x_2}$ to $\Phi_{x_1,x_2}$. For all $x_1 \in X_1$, define
$\chi_{x_1}$ as follows: $\chi_{x_1}= \liand_{x_2 \in X} \delta_{x_1,x_2}$.
\closedefi
\end{definition}
We now prove the following auxiliary lemmas:

\begin{lemma}\label{lem:deltachi}
For all \cm{s} $\model=(X,\closure,\peval)$, $x, x_1$ and $x_2 \in X$, the following holds:
\begin{enumerate}
\item\label{lem:deltachi:1}  $\model,x \models \chi_{x}$;
\item\label{lem:deltachi:2}  $\model,x_2 \models \chi_{x_1}$ if and only if $x_1 \imleq x_2$.
\end{enumerate}
\end{lemma}

\begin{proof}
The assert follows directly from the relevant definitions.
\end{proof}

\begin{lemma}\label{lem:SsubOR}
For all \cm{s} $\model=(X,\closure,\peval)$ and $S \subseteq X$, the following holds:
$S \subseteq \sem{\lior{s\in S} \chi_{s}}$.
\end{lemma}

\begin{proof}
$ $\\
\noindent
$
\deriv
y \in S
\hint{\Rightarrow}{$S \subseteq X$ and Lemma~\ref{lem:deltachi}(\ref{lem:deltachi:1})}
\model,y \models \chi_y
\hint{\Leftrightarrow}{Definition of $\sem{\cdot}$}
y\in \sem{\chi_y}
\hint{\Rightarrow}{$y\in S$}
y\in \bigcup_{s \in S}\sem{\chi_{s}}
\hint{\Leftrightarrow}{$\bigcup_{s \in S}\sem{\chi_{s}}=
\sem{\bigvee_{s \in S}\chi_{s}}$}
y\in\sem{\bigvee_{s \in S}\chi_{s}}
$
\end{proof}

We now proceed with the proof of the theorem.
Assume $x_1 \imleq x_2$. Condition 1 of Definition~\ref{def:CMbisimilarity} is trivially satisfied. Let us consider Condition 2. Let $S_1 \subseteq X$ be any set such that 
$x_1 \in \interior(S_1)$. Since $x_1 \in \interior(S_1)$ and, by Lemma~\ref{lem:SsubOR} below, 
$S_1 \subseteq \sem{\lior_{s_1\in S_1} \chi_{s_1}}$, by monotonicity of $\interior$ we get 
$x_1 \in\interior(\sem{\lior_{s_1\in S_1} \chi_{s_1}})= 
\overline{\closure(\overline{\sem{\lior_{s_1\in S_1} \chi_{s_1}}})}$.
This means that $x_1 \not\in \closure(\overline{\sem{\lior_{s_1\in S_1} \chi_{s_1}}})$, and so we get 
$\model, x_1 \not\models \lnear (\lneg \lior_{s_1\in S_1} \chi_{s_1})$.
Thus we have
$\model, x_1 \models \lneg \lnear (\lneg \lior_{s_1\in S_1} \chi_{s_1})$.
Since $x_1 \imleq x_2$, we have that also
$\model, x_2 \models \lneg \lnear (\lneg \lior_{s_1\in S_1} \chi_{s_1})$ holds, which means that  $x_2 \in \interior(\sem{\bigvee_{s_1\in S_1} \chi_{s_1}})$. 
Take now $S_2 = \sem{\lior_{s_1\in S_1} \chi_{s_1}}$. 
Let $s_2$ be any element of $S_2$. By definition of $S_2$ there exists $s_1 \in S_1$ such that $\model, s_2 \models \chi_{s_1}$, that means, by Lemma~\ref{lem:deltachi}(\ref{lem:deltachi:2}), $s_1 \imleq s_2$. The proof for Condition 3 is similar, using symmetry.

\section{Proofs of Results of Section~\ref{sec:CMCbisimilarity}}\label{apx:sec:CMCbisimilarity}

\begin{remark}\label{rem:prp:CMCbisImpliesCMbis}
The converse of \propo{}~\ref{prp:CMCbisImpliesCMbis} does not hold
as shown in Figure~\ref{fig:CMbisNoImplCMCbis} where 
$\invpeval(u_{11})=\invpeval(u_{21})=\SET{r}, \invpeval(u_{12})=\invpeval(u_{22})=\SET{g}$ and $\invpeval(u_{13})=\SET{b,g}$.
It is easy to see that $\SET{(u_{11},u_{21}),(u_{12},u_{22})}$ is a \cm-bisimulation whereas there is no \cmc-bisimulation $B$ containing $(u_{11},u_{21})$; in fact,
any such relation should satisfy Condition (4) of Definition~\ref{def:CMCbisimilarity} for $S_1=\SET{u_{11},u_{12}}$, for which there is only one $S_2$ with $u_{21} \in \interiorT(S_2)$, namely $\SET{u_{21}, u_{22}}$, and this would require 
$(u_{11},u_{22}) \in B$ or $(u_{12},u_{22}) \in B$. But $(u_{11},u_{22}) \in B$ cannot hold
because $\invpeval(u_{11})=\SET{r}\not=\SET{g}=\invpeval(u_{22})$, which would violate Condition 1 of Definition~\ref{def:CMCbisimilarity}. Also  $(u_{12},u_{22}) \in B$ cannot hold because Condition 5 would be violated: take $S_2 =\SET{u_{22}}$ and consider all
sets $S$ such that $u_{12} \in \interiorT(S)$. Any such $S$ would necessarily contain also $u_{13}$ and there is no $s_2\in \SET{u_{22}} = S_2$ such that $(u_{13},s_2)\in B$ and this is because $\peval(b)=\SET{u_{13}}$.
\end{remark}

\subsection{Proof of \propo{} \ref{prp:CMCbisImpliesTraceEq}}\label{prf:prp:CMCbisImpliesTraceEq}
In the sequel, we also exploit the fact that $x_1\, \cmcbis \, x_2$ if and only if
with $x_1\, \clbis \, x_2$ (see Definition~\ref{def:Clbisimilarity} and Corollary~\ref{cor:ClbisEqCMCbis}).
By $x_1\, \clbis \, x_2$ we know there exists \cl-bisimulation $B$ such that 
$(x_1,x_2)\in B$, which implies that $\invpeval(x_1)=\invpeval(x_2)$ by 
Condition 1 of Definition~\ref{def:Clbisimilarity}. 
Let $\theta$ be any element of $\T(\bpthsF(x_1))$ and 
$\pi_1 \in \bpthsF(x_1)$ such that $\theta_1 = \T(\pi_1)$, with $\pthlen(\pi_1)=n$. 
By the \propo{}~\ref{prp:FT}(\ref{path})  we know
that $\pi_1(i) \in \, \closureF(\pi_1(i-1))$ for $i=1,\ldots,n$. We build path $\pi_2 \in \bpthsF(x_2)$ as follows: we let $\pi_2(0) = x_2$;
since $(x_1,x_2)\in B$ and $\pi_1(1) \in \, \closureF(\pi_1(0))$,
we know that  there is an element, say $\eta_1 \in \, \closureF(\pi_2(0))$ 
such that $(\pi_1(1),\eta_1)\in B$: we let $\pi_2(1) = \eta_1$, observing that
$\invpeval(\pi_1(1))=\invpeval(\pi_2(1))$ since $(\pi_1(1), \pi_2(1))\in B$.
With similar reasoning, exploiting \propo{}~\ref{prp:FT}(\ref{path}), we define
$\pi_2(i)$ for $i=2,\ldots,n$ and we let $\dom \, \pi_2 = \SET{0,\ldots n} =\dom \, \pi_1$.
\propo{}~\ref{prp:FT}(\ref{path})  ensures that $\pi_2$ is continuous and so it is a path
from $x_2$ and $(\pi_1(i),\pi_2(i)) \in B$ for $i=0,\ldots,n$, so that 
$\T(\pi_2) = \theta$.
The proof for the other cases is similar, using \propo{}~\ref{prp:FT}(\ref{path}).

\begin{remark}\label{rem:prp:CMCbisImpliesTraceEq}
The converse of \propo{}~\ref{prp:CMCbisImpliesTraceEq} does not hold
as shown in Figure~\ref{fig:TraceEqNoImplCMCbis} where 
$
\invpeval(y_{11})=\invpeval(y_{12})=\invpeval(y_{21})=\invpeval(y_{22})=\invpeval(y_{24})=\SET{r}\not=\SET{b}=\invpeval(y_{13})=\invpeval(y_{23})
$
and $y_{11}  \treq y_{21}$ but $y_{11}  \not\cmcbis y_{21}$.
In fact, recalling again that $\cmcbis$ coincides with $\clbis$
(see Definition~\ref{def:Clbisimilarity} and Corollary~\ref{cor:ClbisEqCMCbis}),
we note that there cannot be any \cl-bisimulation containing 
$(y_{11},y_{21})$ and this is because $y_{24} \in \, \closureF(y_{21})$, with 
$\closureF(y_{24}) =\emptyset$ and $\closureF(y_{11}) =\SET{y_{11},y_{12}}$
and both $\closureF(y_{11})\not=\emptyset$ and $\closureF(y_{12})\not=\emptyset$.
\end{remark}

\subsection{Proof of \theo{} \ref{thm:CMCbisIsImlceq}}\label{prf:thm:CMCbisIsImlceq}
The proof can be carried out by induction on the structure of $\form$. The only interesting cases
are those for $\lnearF$ and $\lnearT$.
The proof for  $\lnearF$ is exactly the same as the proof of Theorem~\ref{thm:CMbisIsImleq} where $B$ is now a \cmc-bisimulation and
$\lnearF, \closureF$, $\interiorF$ and Condition 2 of Definition~\ref{def:CMCbisimilarity}
are used instead of 
$\lnear, \closure{}$, $\interior{}$ and Condition 2 of Definition~\ref{def:CMbisimilarity}.
The proof for  $\lnearT$ is again the same as the proof of Theorem~\ref{thm:CMbisIsImleq} where $B$ is a \cmc-bisimulation and
$\lnearT, \closureT$, $\interiorT$ and Condition 4 of Definition~\ref{def:CMCbisimilarity}
are used instead of 
$\lnear, \closure{}$, $\interior{}$ and Condition 2 of Definition~\ref{def:CMbisimilarity}.

\subsection{Proof of \theo{} \ref{thm:ImlceqIsCMCbis}}\label{prf:thm:ImlceqIsCMCbis}

The proof is exactly the same as the proof of Theorem~\ref{thm:ImleqIsCMbis} where
$\imlceq$ is considered instead of $\imleq$ and,
when proving that the requirements concerning Condition 2 of Definition~\ref{def:CMCbisimilarity} are fulfilled, 
$\lnearF, \closureF$, $\interiorF$ and Condition 2 of Definition~\ref{def:CMCbisimilarity}
are used instead of 
$\lnear, \closure{}$, $\interior{}$ and Condition 2 of Definition~\ref{def:CMbisimilarity}, while for proving that the requirements concerning Condition 4 of Definition~\ref{def:CMCbisimilarity} are fulfilled, 
$\lnearT, \closureT$, $\interiorT$ and Condition 4 of Definition~\ref{def:CMCbisimilarity}
are used instead of 
$\lnear, \closure{}$, $\interior{}$ and Condition 2 of Definition~\ref{def:CMbisimilarity}.
The proof for Condition 3 (Condition 5, respectively) is similar to that of
Condition 2  (Condition 4, respectively), using symmetry.

\subsection{Proof of \theo{} \ref{thm:ClbisIsIslcseq}}\label{prf:thm:ClbisIsIslcseq}
We proceed by induction on the structure of $\form$ and 
consider only the case  $\ltothru \form_1 [\form_2]$, 
the case for $\lfromthru \form_1 [\form_2]$ being similar, and the others being trivial.
Suppose $B$ is a \cl-bisimulation,  $(x_1,x_2) \in B$ and  
$\model,x_1 \models \ltothru \form_1 [\form_2]$. This means that there exist
path $\pi$ and index $\ell$ such that $\pi(0)=x_1$, $\model, \pi(\ell) \models \form_1$
and for all $j \in \ZET{n \in \nats}{0 < n < \ell}$ we have $\pi(j) \models \form_2$.
We define $\pi_1$ as $\pi_1(j)=\pi(j)$ for $j \in \ZET{n \in \nats}{0 \leq n < \ell}$
and $\pi_1(j)=\pi(\ell)$ for $\ell \leq j$.

We build $\pi_2$, such that $\pthlen \,\pi_2=\pthlen \,\pi_1$, as follows. We let 
$\pi_2(0)=x_2$. By Proposition~\ref{prp:FT}(\ref{path}), we know that
$\pi_1(1) \in \closureF(\pi_1(0))$, and since $(\pi_1(0),\pi_2(0))=(x_1,x_2) \in B$
and $B$ is a \cl-bisimulation, we also know that there exists 
$\eta\in \closureF(\pi_2(0))$ such that $(\pi_1(1),\eta) \in B$. 
We let  $\pi_2(1)=\eta$ and we proceed in  a similar way for defining 
$\pi_2(j)\in\closureF(\pi_2(j -1))$ for all $j \leq (\pthlen \, \pi_2) =\ell$, exploiting
Proposition~\ref{prp:FT}(\ref{path}). Again by
Proposition~\ref{prp:FT}(\ref{path}), function $\pi_2$ is continuous and so it is a path.
In addition, since, for all $j \in \ZET{n \in \nats}{0 < n < \ell}$, by hypothesis and construction
we have $\pi_1(j) \models \form_2$ and $(\pi_1(j),\pi_2(j)) \in B$, by
the Induction Hypothesis, we also get $\pi_2(j) \models \form_2$. Similarly,
we get that $\pi_2(\ell) \models \form_1$ since $\pi_1(\ell) \models \form_1$
and $(\pi_1(\ell),\pi_2(\ell)) \in B$. Thus we have that 
$\model,x_2 \models \ltothru \form_1 [\form_2]$ since there is a path $\pi_2$
and index $\ell$ such that $\pi_2(0)=x_2$, $\model, \pi_2(\ell) \models \form_1$
and for all $j \in \ZET{n \in \nats}{0 < n < \ell}$ we have $\pi_2(j) \models \form_2$.

\subsection{Proof of \theo{} \ref{thm:IslcseqIsClbis}}\label{prf:thm:IslcseqIsClbis}
We have to show that Conditions 1-5 of Definition~\ref{def:Clbisimilarity} are satisfied.
We consider only Condition 2, since the proofs for Conditions 3-5 is similar and
Condition 1 is trivially satisfied if $(x_1, x_2) \in \islcseq$. Suppose there exists
$x_1' \in \closureF(\SET{x_1})$ such that $(x_1',x_2')\not\in \islcseq$ for all
$x_2' \in \closureF(\SET{x_2})$. Note that $x_1'\not=x_1$ because $x_2 \in \closureF(\SET{x_2})$ and $(x_1,x_2) \in \islcseq$. Since $(x_1',x_2')\not\in \islcseq$ 
for all $x_2' \in \closureF(\SET{x_2})$, we know that, for each such $x_2'$,
there is a formula $\Phi_{{x_2'}}$ such that, \wlg,
$\model, x_1' \models \Phi_{{x_2'}}$ and $\model, x_2' \not\models \Phi_{{x_2'}}$,
by definition of $\islcseq$. Clearly, we also have
$\model, x_1' \models \liand_{x_2' \in \closureF(\SET{x_2})}\Phi_{{x_2'}}$ and
$\model, x_2' \not\models \liand_{x_2' \in \closureF(\SET{x_2})}\Phi_{{x_2'}}$.
But this brings to 
$\model, x_1 \models \ltothru (\liand_{x_2' \in \closureF(\SET{x_2})}\Phi_{{x_2'}})[\lfalse]$ and
$\model, x_2 \not\models \ltothru (\liand_{x_2' \in \closureF(\SET{x_2})}\Phi_{{x_2'}})[\lfalse]$,
which contradicts  $x_1\, \islcseq \, x_2$.

\subsection{Proof of \lm{} \ref{lem:NearEqRhoFalse}}\label{prf:lem:NearEqRhoFalse}
We prove that $\lnearF \form\equiv  \lfromthru \form[\lfalse]$, the proof for
$\lnearT \form\equiv  \ltothru \form[\lfalse]$ being similar.
We recall that $\model,x \models \lnearF \, \form \Leftrightarrow x \in 
\closureF{\sem{\form}^{\model}}$.
If $\sem{\form}^{\model}= \emptyset$, i.e. if $\form \equiv \lfalse$ then the proposition holds trivially.
So, assume $\sem{\form}^{\model} \not= \emptyset$.
Suppose $\model,x \models \lnearF \form$. We have two cases:\\
{\bf Case 1}: $\model,x \models \form$\\
In this case, take $\pi$ such that $\pi(i)=x$ for all $i \in \nats$. So, there is a path, $\pi$ as above,
such that $\pi(\ell)=x$, for $\ell=0$, $\model,\pi(0) \models \form$, and there is no 
$j\in \nats$ ssuch that $0 < j < 0$; therefore 
$\model,x \models \lfromthru \form[\lfalse]$.\\
\noindent
{\bf Case 2}: $\model,x \not\models \form$\\
In this case, we know  
$x\in \closureF(\sem{\form}^{\model})\setminus \sem{\form}^{\model}) \not=\emptyset$  by definition of $\lnearF \form$
and by hypothesis.
Since, by hypothesis, $\sem{\form}^{\model} \not= \emptyset$, $x \in \closureF(\sem{\form}^{\model})\setminus \sem{\form}^{\model} \not=\emptyset$, and $\closureF(\sem{\form}^{\model})= \cup_{x'\in \sem{\form}^{\model}} \closureF(\SET{x'})$, then there exists $x'\not=x$ with $x' \in \sem{\form}^{\model}$ and $x\in \closureF(\SET{x'})$.
Let $\pi$ be defined as follows 
$\pi(0)=x'$ and $\pi(j)=x$ for all $j\in \nats$ s.t. $j\geq1$;
by Proposition~\ref{prp:FT}(\ref{path}) $\pi$ is a path and so we get 
$\model, x \models \lfromthru \form[\lfalse]$ by definition of $\lfromthru$.

For the the proof of the converse, let us assume 
$\model,x \models \lfromthru \form[\lfalse]$.
This means there exists $\pi$ and $\ell$ such that $\pi(\ell)=x$, $\model, \pi(0) \models \form$
and for all $j \in \nats$ s.t. $0 < j < \ell$ it holds $\model, \pi(j) \models \lfalse$;
obviously there cannot be any such a $j$, which implies that there are only two cases:\\
{\bf Case 1}: $\ell=0$\\
In this case we have $\model, x \models \form$, which implies $x \in \sem{\form}^{\model}$, and thus
$x \in \closureF(\sem{\form}^{\model})$, so that $\model, x \models \lnearF \form$\\
{\bf Case 2}: $\ell=1$\\
From continuity of $\pi$, we get that $x \in \closureF(\SET{\pi(0)})$, as follows:\\\\
\noindent
$
\deriv
x
\hint{=}{By hypothesis}
\pi(1)
\hint{\in}{Set theory}
\SET{\pi(0),\pi(1)}
\hint{=}{Algebra}
\pi(\SET{0,1})
\hint{=}{Definition of $\closureF_{\succ}$}
\pi(\closureF_{\succ}(\SET{0}))
\hint{\subseteq}{Continuity of $\pi$}
\closureF(\pi(\SET{0}))
\hint{=}{Algebra}
\closureF(\SET{\pi(0)})
$\\
So, by monotonicity of $\closureF$, since $\pi(0) \in \sem{\form}^{\model}$, we have $x \in \closureF(\sem{\form}^{\model})$,
that is $\model, x \models \lnearF\form$.

\subsection{Proof of \theo{}\ref{thm:ClbisEqImlceq}}\label{prf:thm:ClbisEqImlceq}
The proof that $x_1 \, \clbis \, x_2$ implies $x_1 \, \imlceq \, x_2$ follows directly
from Theorem~\ref{thm:ClbisIsIslcseq} and  Lemma~\ref{lem:NearEqRhoFalse}.
The proof that $x_1 \, \imlceq \, x_2$ implies $x_1 \, \clbis \, x_2$ is exactly the same 
as that of Theorem~\ref{thm:IslcseqIsClbis} where, $\lnearT (\liand_{x_2' \in \closureF(\SET{x_2})}\Phi_{{x_2'}})$ is used instead of 
$\ltothru (\liand_{x_2' \in \closureF(\SET{x_2})}\Phi_{{x_2'}}) [\lfalse]$ and 
similarly for  $\lnearF$ and  $\lfromthru$. 

\subsection{$\lsurr$ and  $\lprop$ as derived operators}\label{apx:prp:SurrAndPropDerived}

The surrounded and the propagation operators of~\cite{Ci+16} can be derived from the reachability ones $\ltothru$ and $\lfromthru$,
noting that the proposition below is not restricted for \qdcm{} but it holds for {\em general} \cm{.}  We first recall the definition of the surrounded and of the propagation operators as given in~\cite{Ci+16}:

\begin{definition}\label{def:SurrAndProp}
The satisfaction relation for (general) \cm{s}  $\model$, $x \in \model$, and \slcs{} formulas 
$\form_1 \, \lsurr \, \form_2$ and $\form_1 \, \lprop \, \form_2$ is defined recursively on the structure of $\form$ as follows:
$$
\begin{array}{r c l c l c l}
\model,x  & \models_{\slcs} & \form_1 \, \lsurr \, \form_2 & \Leftrightarrow &
\model, x \models_{\slcs} \form_1 \mbox{ and }\\
&&&&\mbox{for all  paths } \pi \mbox{ and indexes } \ell \mbox{ the following holds:}\\
&&&&\mbox{\hspace{0.1in}} \pi(0) = x \mbox{ and } \pi(\ell) \models_{\slcs} \lneg\form_1\\
&&&&\mbox{\hspace{0.1in}} \mbox{implies}\\
&&&&\mbox{\hspace{0.1in}} \mbox{there exists index } j \mbox{ such that:}\\
&&&&\mbox{\hspace{0.2in}} 0 < j \leq \ell \mbox{  and }
\pi(j) \models_{\slcs} \form_2;\\
\model,x  & \models_{\slcs} & \form_1 \, \lprop  \, \form_2 & \Leftrightarrow &
\model, x \models_{\slcs} \form_2 \mbox{ and }\\
&&&&\mbox{there exist  path } \pi \mbox{ and index } \ell \mbox{ such that }\\
&&&&\mbox{\hspace{0.1in}} \pi(\ell) = x \mbox{ and }\\
&&&&\mbox{\hspace{0.1in}} \pi(0) \models_{\slcs} \form_1 \mbox{ and }\\
&&&&\mbox{\hspace{0.1in}} \mbox{for all } j \mbox{ such that } 0 < j <  \ell \mbox{ the following holds:}\\
&&&&\mbox{\hspace{0.2in}} \pi(j) \models_{\slcs} \form_2.
\end{array}
$$
\closedefi
\end{definition}

\begin{proposition}\label{prp:SurrAndPropDerived}
For all \cm{s}  $\model=(X,\closure,\peval)$ the following holds:
\begin{enumerate}
\item\label{Sur}
$
\form_1 \, \lsurr \, \form_2  \equiv 
\form_1 \, \land\lneg 
(\ltothru (\lneg(\form_1 \lor \, \form_2)[\lneg \form_2]);
$
\item\label{Pro}
$
\form_1 \, \lprop \, \form_2  \equiv 
\form_2 \, \land \lfromthru \form_1[\form_2].
$
\end{enumerate}
\end{proposition}

\begin{proof}
For what concerns Proposition~\ref{prp:SurrAndPropDerived}(\ref{Sur})
We prove that 
$$
\model, x \not\models \form_1 \, \lsurr \, \form_2
\mbox{ if and only if }
\model, x \not\models \form_1 \, \land\lneg 
(\ltothru(\lneg(\form_1 \lor \, \form_2))[\lneg \form_2])
$$
by the following derivation:\\

\noindent
$
\deriv
\model, x \not\models \form_1 \, \land\lneg 
(\ltothru(\lneg(\form_1 \lor \, \form_2))[\lneg \form_2])
\hint{\Leftrightarrow}{Defs. of $\not\models$, $\land$; Logic}
\model, x \not\models \form_1 \mbox{ or } \contnl
\model, x \not\models \lneg 
(\ltothru(\lneg(\form_1 \lor \, \form_2))[\lneg \form_2])\\
\hint{\Leftrightarrow}{Defs. of $\not\models$, $\lneg$}
\model, x \not\models \form_1 \mbox{ or }  \contnl
\model, x \models
\ltothru(\lneg(\form_1 \lor \, \form_2))[\lneg \form_2]\\
\hint{\Leftrightarrow}{Definition of $\ltothru\Phi[\Psi]$}
\model, x \not\models \form_1\mbox{ or }  \contnl
\begin{array}{l l l}
\mbox{exist path} &\pi  \mbox{ and index }  \ell  \mbox{ s.t. :}   &\\
& \pi(0)=x \,\mbox{ and }  \\
& \model, \pi(\ell) \models \lneg (\form_1 \lor \, \form_2)\,\mbox{ and } \\
& \mbox{for all } j : 0 < j < \ell \mbox{ implies }\model,\pi(j) \models  \lneg \form_2
\end{array}\\
\hint{\Leftrightarrow}{Defs. of $\lneg$, $\lor$, $\not\models$;  Logic}
\model, x \not\models \form_1 \mbox{ or } \contnl
\begin{array}{l l l}
\mbox{exist path} & \pi  \mbox{ and index } \ell  \mbox{ s.t. :} \\
  & \pi(0)=x \,\mbox{ and } \\
& \model, \pi(\ell) \models \lneg\form_1 \mbox{ and } \\
& \model, \pi(\ell) \models \lneg\form_2 \mbox{ and } \\
& \mbox{for all } j: 0 < j < \ell \mbox{ implies }\model,\pi(j) \models  \lneg\form_2
\end{array}\\
\hint{\Leftrightarrow}{Logic}
\model, x \not\models \form_1 \mbox{ or } \contnl
\begin{array}{l l l}
\mbox{exist path} & \pi  \mbox{ and index } \ell  \mbox{ s.t. :} \\
 & \pi(0)=x \,\mbox{ and } \\
& \model, \pi(\ell) \models \lneg\form_1 \mbox{ and } \\
& \mbox{for all } j : 0< j \leq \ell \mbox{ implies } \model,\pi(j) \models  \lneg \form_2
\end{array}\\
\hint{\Leftrightarrow}{Defs. of $\not\models$, $\lsurr$}
\model, x \not\models \form_1 \, \lsurr \, \form_2
$

\noindent
The proof of Proposition~\ref{prp:SurrAndPropDerived}(\ref{Pro}) trivially follows
from the relevant definitions.
\end{proof}

\section{Proofs of Results of Section~\ref{sec:Pathbisimilarity}}\label{apx:sec:Pathbisimilarity}

\subsection{Proof of \propo{} \ref{prp:CMCbisImplPath}}\label{prf:prp:CMCbisImplPath}

We prove that every relation $B \subseteq X \times X$ that is a \cmc-bisimulation  is also a \pth-bisimulation.

Suppose $(x_1,x_2)\in B$; we have to prove that Conditions 1-5 of Definition~\ref{def:Pathbisimilarity} are satisfied. This is trivially the case for Condition 1, since $(x_1,x_2)\in B$ and $B$ is a \cmc-bisimulation.
Let $\pi_1$ be a path in $\bpthsF(x_1)$ and suppose $\pthlen\,\pi_1 = n >0$, the case $n=0$ being trivial. By the \propo{}~\ref{prp:FT}(\ref{path}) we know
that $\pi_1(i) \in \closureF(\pi_1(i-1))$ for $i=1,\ldots,n$. We build path $\pi_2$ 
as follows: we let $\pi_2(0) = x_2$;
since $(x_1,x_2)\in B$ and $\pi_1(1) \in \closureF(\pi_1(0))$,
we know that  there is an element, say $\eta_1 \in \closureF(\pi_2(0))$ 
such that $(\pi_1(1),\eta_1)\in B$: we let $\pi_2(1) = \eta_1$.
With similar reasoning, exploiting \propo{}~\ref{prp:FT}(\ref{path}), we define
$\pi_2(i)$ for $i=2,\ldots,n$ and we let $\dom \, \pi_2 = \SET{0,\ldots, n} =\dom \, \pi_1$.
Again, \propo{}~\ref{prp:FT}(\ref{path})  ensures that $\pi_2$ is continuous and so it is a path
from $x_2$ and $(\pi_1(n),\pi_2(n)) \in B$.
The proof for the other conditions is similar, using \propo{}~\ref{prp:FT}(\ref{path}).

\begin{remark}\label{rem:prp:CMCbisImplPath}
The converse of  \propo{} \ref{prp:CMCbisImplPath} does not hold, as shown in Figure~\ref{fig:PathNoImplCMCbis} where 
$
\invpeval(x_{11})=\invpeval(x_{21})=\invpeval(x_{22})=\SET{r}\not=\SET{b}=\invpeval(x_{12})=\invpeval(x_{23})
$
and $x_{11}  \pthbis x_{21}$ but $x_{11}  \not\cmcbis x_{21}$. In fact
$B=\SET{(x_{11},x_{21}),(x_{11}, x_{22}),(x_{12},x_{23})}$ is a \pth-Bisimulation.
We show that $x_{11}\not\clbis x_{21}$, i.e. there exists no
\cl-bisimulation containing $x_{11}$ and $x_{21}$;
$x_{12} \in  \closure(\SET{x_{11}})$ and $\invpeval x_{12} = \SET{b}$.
All $x_{2j} \in \closure(\SET{x_{21}})$ are such that $\invpeval x_{2j} = \SET{r}$;
thus, there cannot be any \cl-bisimulation $B$ such that $(x_{12},x_{2j}) \in B$, for $j=1,2$, since Condition 1 of Definition~\ref{def:Clbisimilarity} would be
violated.  Thence there cannot exist any \cl-bisimulation containing $(x_{11},x_{21})$ since Condition 2 of Definition~\ref{def:Clbisimilarity} would be violated. This brings to $x_{11}  \not\clbis x_{21}$, i.e. $x_{11}  \not\cmcbis x_{21}$.
\end{remark}

\subsection{Proof of \propo{} \ref{prp:MppImplInlImplpath}}\label{prf:prp:MppImplInlImplpath}

Suppose $B$ is an \inl-bisimulation and $(x_1,x_2) \in B$. We have to prove that
Conditions 1-5 of Definition~\ref{def:Pathbisimilarity} hold. We prove only Condition 2, the
proof for Conditions 3-5 being similar and that for Condition 1 trivial.
Suppose $x_1 \leadspth{\pi_1} x_1'$. Take neighbourhood $S_1$ of $x_1$ such that
$\rng(\pi_1) \subseteq S_1$---such an $S_1$ exists because $\rng(\pi_1) \subseteq X$ and $X= \interior(X)$. Since $B$ is a \inl-bisimulation, there exists neighbourhood $S_2$ of $x_2$ and $x_2' \in S_2$ such that $(x_1',x_2') \in B$, by Condition~\ref{back} 
of Definition~\ref{def:INLbisimilarity}.
In addition, since $X$ is path-connected, there is $\pi_2$ such that
$x_2 \leadspth{\pi_2} x_2'$.
%

\subsection{Proof of \theo{} \ref{thm:PathbisIsIrleq}}\label{prf:thm:PathbisIsIrleq}
We proceed by induction on the structure of formulas and consider only the case  
$\lto \form$,  the case for $\lfrom \form$ being similar, and the others being trivial.
So, let us assume that for all $x_1,x_2$, if $x_1 \pthbis x_2$, then
$\model, x_1 \models \form$ if and only if $\model, x_2 \models \form$
and prove the assert for $\lto \form$.
Assume $(x_1,x_2)$ is an element  of \pth-bisimulation $B$ and suppose that $\model, x_1 \models \lto \form$.
This means there exist $\pi,\ell$ s.t. $\pi(0)=x_1$ and $\model,\pi(\ell)\models \form$.
So, there is $\pi_1 \in \bpthsF(x_1)$ such that $\model, \pi_1(\pthlen(\pi_1)) \models \form$.
Moreover, since $(x_1,x_2)\in B$, by Condition 2 of the definition of 
\pth-bisimulation, there is also
$\pi_2 \in \bpthsF(x_2)$ such that 
$(\pi_1(\pthlen(\pi_1)), \pi_2(\pthlen(\pi_2))) \in B$.
This, by definition of $\pthbis$, means that we have 
$\pi_1(\pthlen \, \pi_1) \pthbis \pi_2(\pthlen \, \pi_2)$.
By the I.H. we then get
$\model, \pi_2(\pthlen(\pi_2)) \models \form$, from which 
$\model, x_2 \models \lto \form$ follows.\\

\subsection{Proof of \theo{} \ref{thm:IIrleqIsPathbis}}\label{prf:thm:IIrleqIsPathbis}

We have to prove that Conditions 1-5 of Definition~\ref{def:Pathbisimilarity} are fullfilled
by $\irleq$. We consider only Condition 2, since the proof of Conditions 3-5 is similar and that
of Condition 1 is trivial.
Suppose $(x_1,x_2)\in \irleq$ and that 
Condition 2 is not satisfied; this means that there exists 
$\bar{\pi}\in \bpthsF(x_1)$ such that for all $\pi \in \bpthsF(x_2)$   
the following holds: 
$(\bar{\pi}(\pthlen(\bar{\pi})),  \pi(\pthlen(\pi))) \not\in \irleq$. 

For each $\pi \in \bpthsF(x_2)$, let $\Omega_{\pi}$ be a formula such that
$\model, \bar{\pi}(\pthlen(\bar{\pi})) \models \Omega_{\pi}$ and
$\model, \pi(\pthlen(\pi)) \not\models \Omega_{\pi}$---such a formula exists because
$\bar{\pi}(\pthlen(\bar{\pi}))\not\irleq\pi(\pthlen(\pi)) $. 
Clearly,  $\model, \bar{\pi}(\pthlen(\bar{\pi})) \models \bigwedge_{\pi \in \bpthsF(x_2)} \Omega_{\pi}$ and, consequently, we have
$\model, x_1 \models \lto (\bigwedge_{\pi \in \bpthsF(x_2)} \Omega_{\pi})$ whereas 
$\model, x_2 \not\models \lto (\bigwedge_{\pi \in \bpthsF(x_2)} \Omega_{\pi})$, which 
would contradict $(x_1,x_2)\in \irleq$.

\section{Proofs of Results of Section~\ref{sec:CoPabisimilarity}}\label{apx:sec:CoPabisimilarity}

\subsection{Proof of \propo{}~\ref{prp:CMCbisImplCoPa}}\label{prf:prp:CMCbisImplCoPa}
Suppose $x_1 \cmcbis x_2$, i.e. $x_1 \clbis x_2$ (see Corollary~\ref{cor:ClbisEqCMCbis}). Then there exists \cl-bisimulation
$B \subseteq X \times X$ such that $(x_1,x_2) \in B$. By Lemma~\ref{lem:BBrst} below
we know that $B^{rst}\subseteq X \times X$ is a \cop-bisimulation and since 
$B \subseteq B^{rst}$ we have $(x_1,x_2) \in B^{rst}$, i.e. $x_1 \copbis x_2$.

\begin{lemma}\label{lem:BBrst}
For all  \qdcm{s} $(X,\closureF,\peval)$ and relations $B \subseteq X \times X$ the following holds: if $B$ is a \cl-bisimulation, then $B^{rst}$ is a \cop-bisimulation.
\end{lemma}

\begin{proof}
We have to prove that $B^{rst}$ satisfies Conditions 1-5 of Definition~\ref{def:CoPabisimilarity},
under the assumption that $B$ is a \cl-bisimulation. We consider only Condition 1 and Condition 2, since the proof for all the other conditions is similar.
Suppose $(x_1,x_2)\in B^{rst}$. 
For what concerns Condition 1 there are four cases to consider:
\begin{enumerate}
\item $x_1=x_2$: trivial; 
\item $(x_1,x_2) \in B$: in this case $\invpeval x_1 = \invpeval x_2$ since $B$ is a \cl-bisimulation; 
\item $(x_1,x_2) \in B^s \setminus B$: in this case $(x_2,x_1) \in B$---by definition of $B^s$, and so $\invpeval x_2 = \invpeval x_1$;
\item there are $y_1, \ldots, y_n \in X$ such that $y_1=x_1$, $y_n=x_2$ and for all
$i\in\SET{1,\ldots,n-1}$  we have $(y_i,y_{i+1})\in B^s$: in this case 
$\invpeval y_i = \invpeval y_{i+1}$ for all $i\in\SET{1,\ldots,n-1}$---see cases (2) and (3) above---and so also $\invpeval x_1 = \invpeval x_2$.
\end{enumerate}

For what concerns Condition 2, let $\pi_1$ any path in $\bpthsF(x_1)$ such that
$(\pi_1(i_1),x_2) \in B^{rst}$ for all $i_1  < \pthlen(\pi_1)$, and assume 
$\pthlen(\pi_1) > 0$---the case $\pthlen(\pi_1) = 0$ being trivial by choosing $\pi_2$ such that $\pi(i_2)=x_2$ for all $i_2$. By \propo{}~\ref{prp:FT}(\ref{path})  we know that $\pi_1(i_1) \in \closureF(\pi_1(i_1 - 1))$  for all $i_1=1,\ldots, \pthlen(\pi_1)$. We build $\pi_2$, such that $\pthlen(\pi_2)=\pthlen(\pi_1)$, as follows. We let $\pi_2(0)=x_2$; since $(\pi_1(0),\pi_2(0))=(x_1,x_2) \in B^{rst}$ and 
$\pi_1(1) \in \closureF(\pi_1(0))$, there is, by Lemma~\ref{lem:BstrExtB}, $\eta\in \closureF(\pi_2(0))$ s.t. $(\pi_1(0),\eta) \in B^{rst}$. We let  $\pi_2(1)=\eta$ and we proceed in  a similar way for defining $\pi_2(i_2)\in\closureF(\pi_2(i_2 -1))$ for all
$i_2 < \pthlen(\pi_2)$, ensuring that for all such $i_2$, $(\pi_1(0),\pi_2(i_2)) \in B^{rst}$.

Now, by hypothesis and since $\pi_2(0)=x_2$ by definition, we know that\\
$(\pi_1(\pthlen(\pi_1)-1),\pi_2(0)) \in B^{rst}$ and
$(\pi_1(0),\pi_2(0)) \in B^{rst}$, and, by symmetry of $B^{rst}$, also $(\pi_2(0),\pi_1(0)) \in B^{rst}$.
By construction of $\pi_2$, we have also $(\pi_1(0),\pi_2(\pthlen(\pi_2)-1)) \in B^{rst}$.
Thence, by transitivity of $B^{rst}$, we finally get 
$(\pi_1(\pthlen(\pi_1)-1),\pi_2(\pthlen(\pi_2)-1)) \in B^{rst}$.
But then, by \propo{}~\ref{prp:FT}(\ref{path})  we know that 
$\pi_1(\pthlen(\pi_1)) \in \closureF(\pi_1(\pthlen(\pi_1) - 1))$ and so, again by 
Lemma~\ref{lem:BstrExtB}, we know that there exists $\xi \in \closureF(\pi_2(\pthlen(\pi_2)-1))$
such that $(\pi_1(\pthlen(\pi_1)),\xi) \in B^{rst}$. We define  
$\pi_2(\pthlen(\pi_2)) = \xi$; so $(\pi_1(\pthlen(\pi_1)),\pi_2(\pthlen(\pi_2))) \in B^{rst}$ and, 
noting that, again by \propo{}~\ref{prp:FT}(\ref{path}), 
the resulting function $\pi_2$ is continuous, i.e. it is a path, we get the assert.
\end{proof}

\begin{lemma}\label{lem:BstrExtB}
For all \qdcm{s} $\model=(X,\closureF, \peval)$, \cl-bisimulation $B$ and
$(x_1,x_2) \in B^{rst}$ the following holds:
for all $x_1' \in \closureF(x_1)$ there exists $x_2' \in \closureF(x_2)$ such that
$(x_1',x_2') \in B^{rst}$
\end{lemma}

\begin{proof}
There are four cases to consider:
\begin{enumerate}
\item $x_1=x_2$: trivial; 
\item $(x_1,x_2) \in B$: in this case the assert follows directly from the fact that $B$ is a \cl-bisimulation and $B \subseteq B^{rst}$; 
\item $(x_1,x_2) \in B^s \setminus B$: in this case
$(x_2,x_1) \in B$, by Definition of $B^s$, and since $B$ is a \cl-bisimulation, by Condition 3 of Definition~\ref{def:Clbisimilarity}, 
for all $x_1' \in \closureF(x_1)$ there exists $x_2' \in \closureF(x_2)$ such that 
$(x_2',x_1') \in B$; this means that is $(x_1',x_2') \in B^s \subseteq B^{rst}$;
\item there are $y_1, \ldots, y_n \in X$ such that $y_1=x_1$, $y_n=x_2$ and for all
$i\in\SET{1,\ldots,n-1}$  we have $(y_i,y_{i+1})\in B^s$: in this case---by applying the
same reasoning as for cases (2) and (3) above---we have that
for all $y_i' \in \closureF(y_i)$ there is $y_{i+1}' \in \closureF(y_{i+1})$ with
$(y_i', y_{i+1}') \in B^s\subseteq B^{rst}$, for all $i\in\SET{1,\ldots,n-1}$; the assert then 
follows by transitivity of $B^{rst}$.
\end{enumerate}
\end{proof}

\subsection{Proof of \propo{} \ref{prp:DbsEQCopa}}\label{prf:prp:DbsEQCopa}
In the sequel, for the sake of readability, we will let $\rightarrow$ denote the transition relation of $K(\model)^r$, i.e. $\rightarrow=R^r$.

We prove  that $x_1 \, \copbis \, x_2$ implies $x_1 \, \dbsceq \, x_2$ by showing that
$\copbis$ is a \dbsc-Eq w.r.t. $K(\model)^r$.
We know  that Condition (1) of Definition~\ref{def:dbsc} is trivially satisfied since $x_1 \, \copbis \, x_2$.
For what concerns Condition (2) of Definition~\ref{def:dbsc}, suppose
$x_1 \rightarrow x_1'$ in $K(\model)^r$; this means that there $\pi_1 \in \bpthsF_{\nats,\model}(x_1)$ with $\pthlen(\pi_1)=1$ and $\pi_1(1) = x_1'$ and $\pi_1(0) \copbis \, x_2$. But then, since $\copbis$ is a \cop-bisimulation, there is 
$\pi_2 \in \bpthsF_{\nats,\model}(x_2)$ such that $x_1 \copbis \pi_2(i)$ for all
$i<\pthlen(\pi_2)$ and $\pi_1(\pthlen(\pi_1)) \copbis \pi_2(\pthlen(\pi_2))$. This in turn means that there exist $k \in \nats$, $k=\pthlen(\pi_2)$, and 
$t_0=\pi_2(0), \ldots , t_k = \pi_2(k)\in X$ such that 
$x_2 = t_0, x_1 \,\copbis \, t_i , t_i \,\rightarrow \, t_{i+1}$ for all $i < k$, and  $x_1' = \pi_1(\pthlen(\pi_1)) \copbis \, \pi_2(\pthlen(\pi_2))= t_k$, due to the definition of
$K(\model)^r$ and to its relationship to $\model$. The proof for Condition (3) of Definition~\ref{def:dbsc} is similar.

Now we prove  that $x_1 \, \dbsceq \, x_2$ implies $x_1 \, \copbis \, x_2$ and we do it by showing that $\dbsceq$ is a \cop-bisimulation (see example in Figure~\ref{fig:DbsIsCopa}).
\begin{figure}
\centerline{\resizebox{2.8in}{!}{\includegraphics{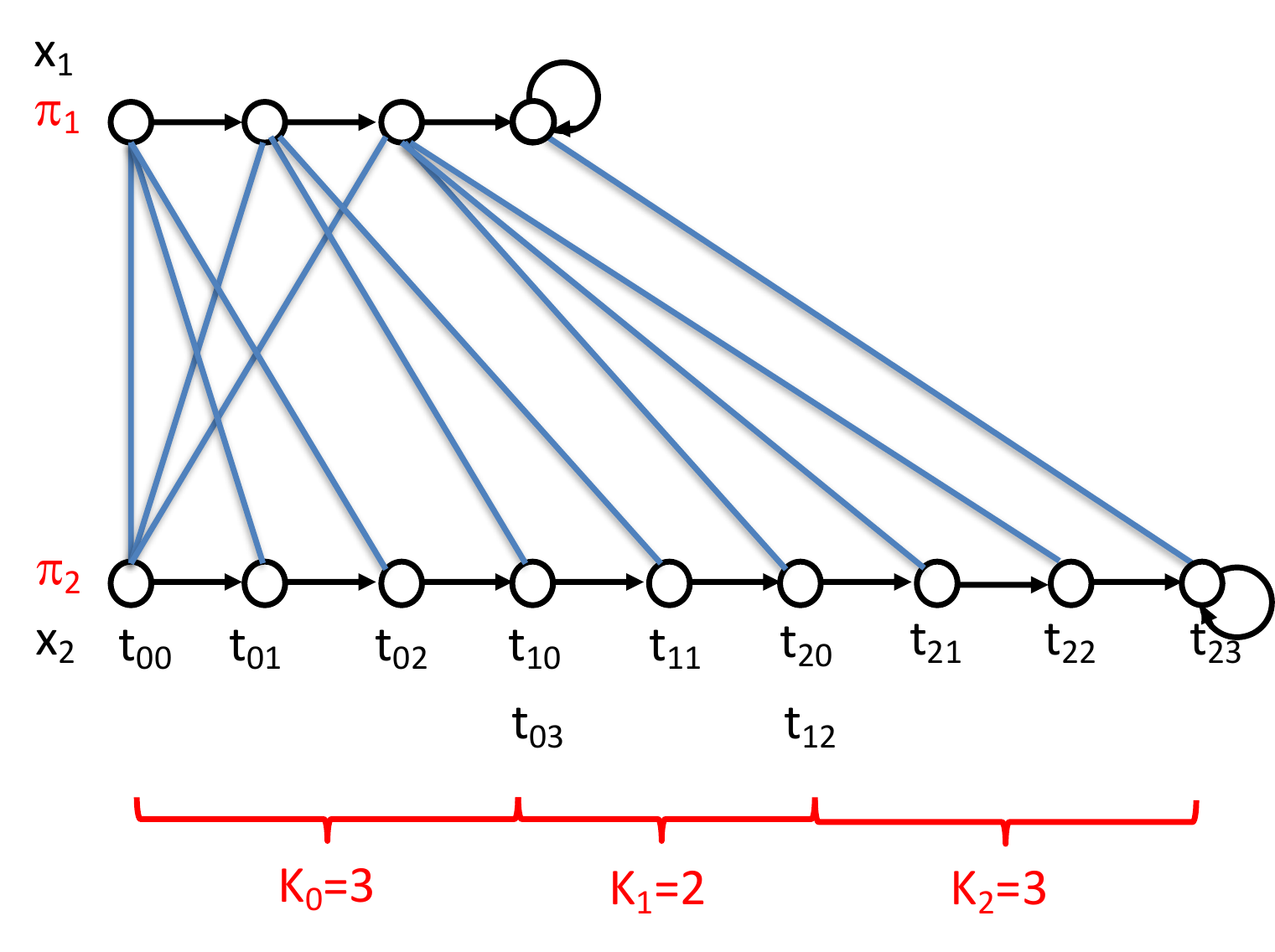}}}
\caption{\label{fig:DbsIsCopa} Example of schema for the proof of \propo{}
 \ref{prp:DbsEQCopa} with $\pthlen(\pi_1)=3, k_0=3, k_1=2$ and $k_2=3$; only ``terminal'' self-loops are shown; $\dbseq$ is shown as blue segments (transitivity of $\dbseq$ is implicit and not shown in the figure).}
\end{figure}

Condition (1) of Definition~\ref{def:CoPabisimilarity} is trivially satisfied because $x_1 \, \dbseq \, x_2$. We prove that also Condition (2) is satisfied, the proof of the remaining conditions being similar. Let  $\pi_1 \in \bpthsF_{\nats,\model}(x_1)$ be any 
path in $\model$ such that 
$\pi_1(i_1)\, \dbsceq \, x_2$ for all $i_1 \in \ZET{\iota}{0 \leq \iota < \pthlen(\pi_1)}$.
We first observe that, due to the definition of $K(\model)^r$ and to its relationship to $\model$, $\pi_1(j)\rightarrow\pi_1(j+1)$, for $j=0, \ldots, \pthlen(\pi_1) - 1$. 
So, for all such $j$ we have that there exist $k_j \in \nats$ and
$t_{j0}, \ldots , t_{jk_j}$ such that $t_{00}=x_2$,
$\pi_1(j) \, \dbsceq \, t_{jm}$ and  $t_{jm} \rightarrow t_{j(m+1)}$ for all $m < k_j$ and 
$\pi_1(j+1) \dbsceq t_{jk_j} = t_{(j+1)0}$; 
Clearly, letting $\ell=\pthlen(\pi_1) - 1$, we have that 
$
t_{00} \rightarrow \ldots \rightarrow t_{0k_0} = 
t_{10} \rightarrow \ldots \rightarrow \ldots \rightarrow 
t_{\ell0} \ldots \rightarrow t_{\ell k_{\ell}}
$
is a path over $K(\model)^r$. Such a path corresponds to the following path $\pi_2$ of 
$\model$, where we let $h(n,j)= n - \sum_{i=0}^{j-1} k_i$ and we assume
$\sum_{i=0}^{w} k_i = 0$ if $w<0$:
$$
\pi_2(n)=
\left\{
\begin{array}{l}
t_{j(h(n,j))}, \mbox{ if there is } j \mbox{ s.t. } 0\leq j \leq \ell \mbox{ and }
\sum_{i=0}^{j-1} k_i \leq n < \sum_{i=0}^{j} k_i,\\\\
t_{\ell k_{\ell}}, \mbox{ if } n \geq \sum_{i=0}^{\ell} k_i.
\end{array}
\right.
$$
Note that $\pthlen(\pi_2) =  \sum_{i=0}^{\ell} k_i$ and that, by construction,
$\pi_1(\pthlen(\pi_1)) \dbsceq \pi_2(\pthlen(\pi_2))$. Note furthermore that
$x_1\,\dbseq \pi_2(i_2)$ for all $i_2 \in \ZET{\iota}{0 \leq \iota < \pthlen(\pi_2)}$. In fact, again by construction, for each $i_2 \in \ZET{\iota}{0 \leq \iota < \pthlen(\pi_2)}$ there is
$i_1 \in \ZET{\iota}{0 \leq \iota < \pthlen(\pi_1)}$ such that $\pi_1(i_1) \, \dbsceq \, \pi_2(i_2)$; moreover, $\pi_1(i_1) \, \dbsceq \,x_2 \dbseq \,x_1$ holds for all such $\pi_1(i_1)$ by hypothesis and so, by transitivity, we also get  $x_1\,\dbsceq \pi_2(i_2)$.

\subsection{Proof of \theo{} \ref{thm:CoPabisIsIcrleq}}\label{prf:thm:CoPabisIsIcrleq}

We proceed by induction on the structure of formulas and consider only the case  
$\lstothru \form_1[\form_2]$,  the case for $\lsfromthru \form_1[\form_2]$ being similar, 
and the others being trivial.
So, let us assume that for all $x_1,x_2$, if $x_1 \copbis x_2$, then
$\model, x_1 \models \form$ if and only if $\model, x_2 \models \form$
and prove the assert for $\lstothru \form_1[\form_2]$.

Suppose that $\model, x_1 \models \lstothru \form_1[\form_2]$.
This means there exist $\pi,\ell$ s.t. $\pi(0)=x_1, \model,\pi(\ell)\models \form_1$
and, for $j \in\ZET{\iota}{0\leq \iota < \ell}$ we have $\model,\pi(j)\models \form_2$.
If $\ell=0$, then, by definition of $\lstothru$, we know that
$\model, x_1 \models \form_1$ and $\model, x_1 \models \form_2$ and, by the I.H.
we  get that also
$\model, x_2 \models \form_1$ and $\model, x_2 \models \form_2$ and, again by
definition of $\lstothru$ we get
$\model, x_2 \models \lstothru \form_1[\form_2]$.
Suppose now that $\ell>0$, and let path $\pi_1$ be defined as follows:
$$
\pi_1(i_1)=
\left\{
\begin{array}{l}
\pi(i_1), \mbox{ if } i_1 \leq \ell,\\
\pi(\ell), \mbox{ if } i_1 > \ell.
\end{array}
\right.
$$

Clearly, $\pi_1 \in \bpthsF(x_1)$, $\pthlen(\pi_1) = \ell$, 
$\model, \pi(\pthlen(\pi_1))\models \form_1$ 
and, for $j \in\ZET{\iota}{0\leq \iota < \pthlen(\pi_1)}$ we have $\model,\pi_1(j)\models \form_2$.
%
Let $B$ be a \cop-bisimulation such that $(x_1,x_2)\in B$; such a $B$
exists since $x_1 \copbis x_2$.
In the sequel, we  will construct a path $\pi_2 \in \bpthsF(x_2)$ such that
$\pi_2(0)=x_2$ and we also have $\model, \pi_2(\pthlen(\pi_2)) \models \form_1$ and
for all $i_2 \in \ZET{\iota}{0 \leq \iota < \pthlen(\pi_2)}$ we have
$\model,\pi_2(i_2) \models \form_2$ thus showing that 
$\model, x_2 \models \lstothru \form_1[\form_2]$ (see Figure~\ref{fig:ProofSchema}).

\begin{figure}
\centerline{\resizebox{4in}{!}{\includegraphics{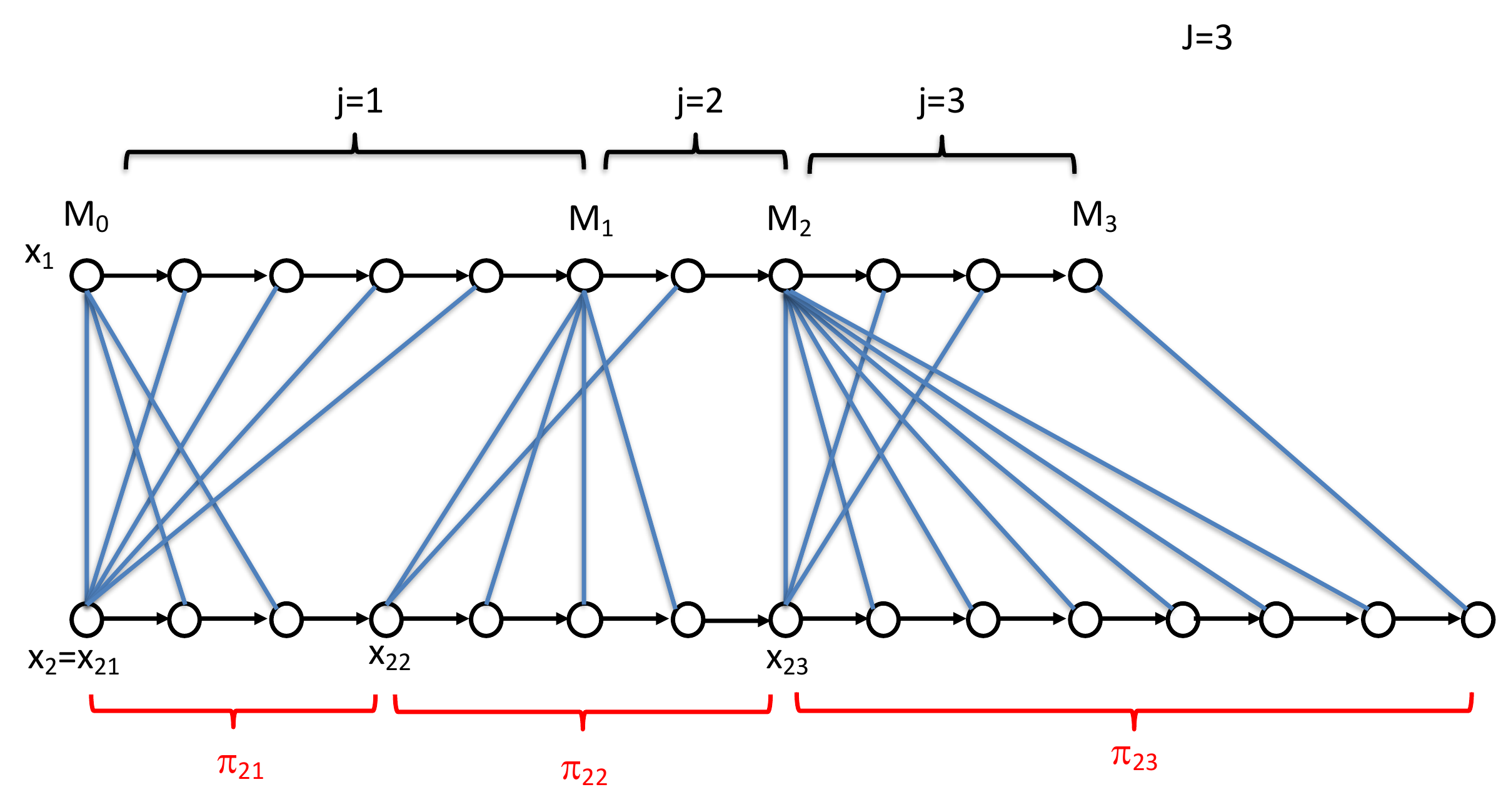}}}
\caption{\label{fig:ProofSchema} Example of schema for the Proof of \theo{} \ref{thm:CoPabisIsIcrleq}, for $J=3$. Relation $B$ is shown as blue segments.}
\end{figure}

Let $M_0=0$, $x_{21}=x_2$.
Now let $M_1$ be the greatest $m_1$ such that  $m_1 \leq \pthlen(\pi_1)$ and
$(\pi_1(i_1), x_{21})\in B$ for all $i_1 \in \ZET{\iota}{M_0 \leq \iota < m_1}$, 
recalling that 
$(\pi_1(M_0), x_{21})\in B$ by hypothesis.
Moreover, since $(\pi_1(M_0), x_{21})\in B$  and 
$B$ is a \cop-bisimulation, by Condition 2 of Definition~\ref{def:CoPabisimilarity},
there exists $\pi_{21}\in \bpthsF(x_{21})$ such that
$(\pi_1(M_0),\pi_{21}(i_2)) \in B$ for all
$i_2 \in \ZET{\iota}{0 \leq \iota < \pthlen(\pi_{21})}$
and $(\pi_1(M_1), \pi_{21}(x_{22}))\in B$, 
where $x_{22}=\pi_{21}(\pthlen(\pi_{21}))$.
Furthermore, since $\model, \pi_1(M_0) \models \form_2$,
by the I.H. we get that also $\model, \pi_{21}(i_2) \models \form_2$ 
for all $i_2 \in \ZET{\iota}{0 \leq \iota < \pthlen(\pi_{21})}$.

For $j>1$, let $M_j$ be the greatest $m_j$ such that  $m_j \leq \pthlen(\pi_1)$ and
$(\pi_1(i_1), x_{2j})$ $\in$  $B$ for all $i_1 \in \ZET{\iota}{Z_{j-1} \leq \iota < z_j}$
recalling that
$(\pi_1(M_{j-1}), x_{2j})\in B$ by definition of $\pi_{2j-1}$.
Moreover, since $(\pi_1(M_{j-1}), x_{2j})\in B$  and 
$B$ is a \cop-bisimulation, by Condition 2 of Definition~\ref{def:CoPabisimilarity},
there exists $\pi_{2j}\in \bpthsF(x_{2j})$ such that 
$(\pi_1(M_{j-1}),\pi_{2j}(i_2)) \in B$ for all
$i_2 \in \ZET{\iota}{0 \leq \iota < \pthlen(\pi_{2j})}$
and $(\pi_1(M_j), \pi_{2j}(x_{2(j+1)}))\in B$, 
where $x_{2(j+1)}=\pi_{2j}(\pthlen(\pi_{2j}))$.
Furthermore, since $\model, \pi_1(M_{j-1}) \models \form_2$,
by the I.H. we get that also $\model, \pi_{2j}(i_2) \models \form_2$ 
for all $i_2 \in \ZET{\iota}{0 \leq \iota < \pthlen(\pi_{2j})}$.

Finally, letting $J$ be the greatest $j$ as above,
since $\model, \pi_1(M_J) \models \form_1$,
by the I.H. we get that also $\model, \pi_{2J}(\pthlen(\pi_{2J})) \models \form_1$. 

We note that $\pi_{2j}(0)=\pi_{2(j-1)}(\pthlen(\pi_{2(j-1)}))$ for $j=1\ldots J$.
Thus we can build the following path $\pi_2$:

$$
\pi_2(n)=
\left\{
\begin{array}{l}
\pi_{21}(n), \mbox{ if } n \in 
[0,\pthlen(\pi_{21})),\\
\vdots\\
\pi_{2j}(n- \sum_{i=1}^{j-1}\pthlen(\pi_{2i})), 
\mbox{ if } n \in 
\left[
\sum_{i=1}^{j-1}\pthlen(\pi_{2i}), 
\sum_{i=1}^{j}\pthlen(\pi_{2i})\right),\\
\vdots\\
\pi_{2J}(n- \sum_{i=1}^{J-1}\pthlen(\pi_{2i})), 
\mbox{ if } n \geq \sum_{i=1}^j \pthlen(\pi_{2i}).
\end{array}
\right.
$$

Clearly, $\pi_2 \in \bpthsF(x_2)$ since $\pi_2(0)=\pi_{2,1}(0)= x_2$ because
$\pi_{21} \in \bpthsF(x_2)$ and $\pi_{2J}$ is bounded. Moreover, by construction,
$\model,\pi_2(i_2) \models \form_2$ for all $i_2 \in \ZET{\iota}{0\leq \iota < \pthlen(\pi_2)}$ and
$\model,\pi_2(\pthlen(\pi_2)) \models \form_1$. 
Thus $\model, x_2 \models \lstothru \form_1[\form_2]$.

\subsection{Proof of \theo{} \ref{thm:IcrleqIsCoPabis}}\label{prf:thm:IcrleqIsCoPabis}

We have to prove that
Conditions 1-5 of the Definition~\ref{def:CoPabisimilarity} are fullfilled. We consider only Condition 2, since the proof for Conditions 3-5 is similar and that of Condition 1 is trivial.
We proceed by contradition. 
Suppose  Condition 2 is not satisfied; this means that there exists 
$\bar{\pi}\in \bpthsF(x_1)$ such that
$(\bar{\pi}(i),x_2) \in \icrleq$ for all $i \in \ZET{\iota}{0\leq \iota < \pthlen(\bar{\pi})}$ and,
for all $\pi \in \bpthsF(x_2)$, having considered that  $\pi(0)=x_2 \,\icrleq\, x_1$, 
the following holds: \\$(\bar{\pi}(\pthlen(\bar{\pi})),  \pi(\pthlen(\pi))) \not\in \icrleq$
or there exists $h_{\pi}$ such that $0 < h_{\pi} < \pthlen(\pi)$ and  $(x_1, \pi(h_{\pi})) \not\in \icrleq$. 
Let  set $I$ be defined as\\

$I = \ZET{\pi \in \bpthsF(x_2)}{\mbox{there exists } h_{\pi} \mbox{ such that } 0 < h_{\pi} < \pthlen(\pi)$\\
\mbox{ }\hspace{3in}$\mbox{ and } (x_1, \pi(h_{\pi})) \not\in \icrleq}$\\

and, for each $\pi \in I$, let 
$\Omega^I_{\pi}$ be a formula such that
$\model, x_1 \models \Omega^I_{\pi}$ and
$\model, \pi(h_{\pi}) \not\models \Omega^I_{\pi}$---such a formula exists because
$(x_1,  \pi(h_{\pi})) \not\in \icrleq$.

Let furthermore set $L$ be defined as   
$$
L=\ZET{\pi \in \bpthsF(x_2)}{(\bar{\pi}(\pthlen(\bar{\pi})),  \pi(\pthlen(\pi))) \not\in \icrleq}
$$ 
and, for each $\pi \in L$, let $\Omega^L_{\pi}$ be a formula such that
$\model, \bar{\pi}(\pthlen(\bar{\pi})) \models \Omega^L_{\pi}$ and
$\model, \pi(\pthlen(\pi)) \not\models \Omega^L_{\pi}$---such a formula exists because
$(\bar{\pi}(\pthlen(\bar{\pi})),  \pi(\pthlen(\pi))) \not\in \icrleq$. Note that
$I \cup L = \bpthsF(x_2)$ by hypothesis.
Clearly, $\model, x_1 \models \bigwedge_{\pi \in I} \Omega^I_{\pi}$ and, since
$(\bar{\pi}(i),x_2) \in \icrleq$ for all $i \in \ZET{\iota}{0\leq \iota <\pthlen(\bar{\pi})}$,
we also get $\model, \bar{\pi}(i) \models \bigwedge_{\pi \in I} \Omega^I_{\pi}$ for
all $i \in \ZET{\iota}{0\leq \iota <\pthlen(\bar{\pi})}$---recall that $\bar{\pi}(0)=x_1$.
Also, $\model, \bar{\pi}(\pthlen(\bar{\pi})) \models \bigwedge_{\pi \in L} \Omega^L_{\pi}$

Thus, we get $\model, x_1 \models \Psi$, where $\Psi$ is  the formula 
$\lstothru (\bigwedge_{\pi \in L} \Omega^L_{\pi}) [\bigwedge_{\pi \in I} \Omega^I_{\pi}]$. 
On the other hand,
$\model, x_2 \not\models \Psi$, since, for every path $\pi \in \bpthsF(x_2)$, 
$\pi(h_{\pi})$ does not satisfy $\bigwedge_{\pi \in I} \Omega^I_{\pi}$ for some
$h_{\pi}$ with $0 < h_{\pi} < \pthlen(\pi)$---by construction 
of $\bigwedge_{\pi \in I} \Omega^I_{\pi}$---or $\pi(\pthlen(\pi))$ does not satisfy 
$\bigwedge_{\pi \in L} \Omega^L_{\pi}$---by construction of $\bigwedge_{\pi \in L} \Omega^L_{\pi}$.
In conclusion, we have found a formula, $\Psi$, such that $\model, x_1 \models \Psi$
whereas $\model, x_2 \not\models \Psi$ and this contradicts $x_1 \, \icrleq\, x_2$.

\end{document}